\documentclass[a4paper,11pt]{article} 
\usepackage{mlmodern}
\usepackage{math}
\usepackage[final]{microtype}
\usepackage{amsmath}
\usepackage{amsthm}
\usepackage{graphicx}
\usepackage[all]{xy}
\usepackage{names}
\usepackage{diagram}
\usepackage{xparse}
\usepackage{expl3}
\usepackage{preamble}
\usepackage{thmtools}
\usepackage{thm-restate}
\usepackage{mathpartir}
\usepackage[raiselinks=false,colorlinks=true,citecolor=gray,urlcolor=gray,linkcolor=gray,bookmarksopen=true]{hyperref}
\usepackage[capitalise]{cleveref}
\newtheorem{thm}{Theorem}

\newtheorem{ax}{Axiom}
\newtheorem{defn}[thm]{Definition}
\newtheorem{ex}[thm]{Example}

\newtheorem{cor}[thm]{Corollary}
\newtheorem{prop}[thm]{Proposition}
\newtheorem{rem}[thm]{Remark}
\Crefname{ex}{Example}{Examples}
\Crefname{lem}{Lemma}{Lemmas}
\Crefname{thm}{Theorem}{Theorems}
\Crefname{defn}{Definition}{Definitions}
\Crefname{cor}{Corollary}{Corollaries}
\Crefname{prop}{Proposition}{Propositions}
\Crefname{rem}{Remark}{Remarks}
\crefformat{footnote}{#2\footnotemark[#1]#3}

\usepackage{url}
\begin{document}

  \title{Powerdomains and nondeterminism in synthetic domain theory} 						%
  \author{Yue Niu\\
    National Institute of Informatics, Tokyo, Japan\\
    \texttt{yue\_niu@nii.ac.jp}
    \and
    Taro Sekiyama\\
    National Institute of Informatics, SOKENDAI, Tokyo, Japan\\
    \texttt{sekiyama@nii.ac.jp}}	%
   \maketitle
\begin{abstract} 
    \emph{Synthetic domain theory} is an axiomatization of domain theory within a constructive universe of sets such that all definable maps between domains are continuous. In this paper we construct the counterparts to the well-known lower, upper, and convex powerdomains in the setting of synthetic domain theory and prove that they produce computationally adequate denotational models of nondeterminism. By developing the theory of powerdomains in synthetic domain theory, we obtain a nondeterministic metalanguage that directly embeds into dependent type theory, where the latter serves as an expressive logic for reasoning about the metalanguage. Moreover, the computational adequacy results imply that denotational reasoning through the metalanguage may be used to study operational behaviors of actual programs. 
\end{abstract}

\section{Introduction}\label{intro}

As a mathematical theory or subfield develops, there is often a proliferation of similar but technically distinct structures whose details of construction can distract from the subject matter at hand. In the case of domain theory, such growing pains began to be addressed by the flourishing development of \emph{synthetic domain theory}~\cite{hyland:1991} in the early 1990s; more recently, there has been a resurgence of work~\cite{sterling-harper:2022,niu-sterling-harper:2024,matache-moss-staton:2022} using synthetic domain theory to provide elegant semantics for computational phenomena that have been challenging to study analytically.

In this paper we continue the program of synthetic domain theory by investigating \emph{powerdomain constructions} and using them to define the denotational semantics of nondeterministic programs. This addresses an apparent gap in the literature: although powerdomain constructions are very well studied in \emph{analytic} domain theory~\cite{plotkin:1976,hennessy-plotkin:1979,smyth:1976,smyth:1983}, their counterparts have received almost no attention in synthetic domain theory. The only prior work we are aware of is the characterization by Phoa and Taylor~\cite{phoa-taylor:1990} of the order structure of the \emph{Plotkin/convex powerdomain}~\cite{plotkin:1976} in synthetic domain theory. Building on their work, we define the synthetic counterparts of the \emph{Hoare/lower powerdomain}~\cite{milne-milner:1979} and \emph{Smyth/upper powerdomain}~\cite{smyth:1976} and show that all three satisfy suitable computational adequacy properties with respect to standard operational semantics. In the following we recall the basic themes of synthetic domain theory, motivate our approach, and discuss how it relates to existing work in the area.

\subsection{Motivation and main ideas}

\subsubsection{Combining programming and reasoning}

Synthetic domain theory, like its sibling fields of synthetic differential geometry and synthetic homotopy theory, is a particular way of axiomatizing domain theory based on the internal language of an arbitrary topos $\ECat$, \ie{} a constructive set theory or dependent type theory. The basic idea is to isolate a full subcategory of predomains $\CCat \hookrightarrow \ECat$ of the ambient topos that is both (internally) complete and cocomplete, thereby furnishing a semantic domain rich enough to interpret many basic domain-theoretic constructions. Fullness means that \emph{every} function between predomains in the subuniverse is continuous (in a suitable sense), a routine but important side condition that can cause subtle errors in arguments if neglected. Thus synthetic domain theory is not only a conceptual reorganization of the content of analytic domain theory but also an approach to domain theory that is potentially more amenable to formalization. Moreover, because the subuniverse of domains embeds into an ambient mathematical universe of \emph{general} sets, we are immediately equipped with a familiar logic (either in terms of sets or types) that allows us to reason \emph{naively} about domains; for instance, one may identify within the subuniverse those domains that correspond to some class of extant analytic domains~\cite{fiore-plotkin:1996,fiore-rosolini:1997:cpos}. Thus a theory of powerdomains in synthetic domain theory forms the semantic substrate for a framework in which one may simultaneously program with nondeterministic effects and reason about the resulting programs (by means of \eg{} temporal logics).

\subsubsection{Synthetic powerdomains}

Following the ethos of synthetic domain theory, we define the synthetic analogues of the classic powerdomain constructions as free algebras for suitable semilattice theories. Indeed this was carried out by Phoa and Taylor~\cite{phoa-taylor:1990} for the convex powerdomain construction, which may be characterized as the free semilattice in the category of predomains. On the other hand, the theories of the lower and upper powerdomains extend that of the convex powerdomain with the \emph{in}equalities $S \le S \vee T$ and $S \vee T \le S$, respectively~\cite{hennessy-plotkin:1979}. In analytic domain theory such order-enriched theories are sensible because the notion of order is \emph{conceptually prior} to the notion of (pre)domains. In contrast, more work is required to reveal the intrinsic order structure of synthetic (pre)domains. As a first pass, one may axiomatize an \emph{interval} \I{} containing two points $0, 1:\I$ and define the \emph{path relation} on a type $X$, where a path $\alpha : x \leadsto y$ is just a function $f : \I \to X$ such that $f(0) = x$ and $f(1) = y$. However, the path relation need not be a partial order or even transitive in general~\cite{fiore:1997}. This may be rectified by calibrating the class of predomains so that the path relation on a predomain is always a partial order. Accordingly, the synthetic lower and upper powerdomains, should they exist, ought to be given by left adjoints from the category of predomains to the respective categories of semilattice algebras. Indeed, we prove that these left adjoints may be defined by means of \emph{quotient inductive types}~\cite{fiore-pitts-steenkamp:2021}, a generalization of inductive types whose constructors may include equations. In particular, we show that the aforementioned inequalities may be encoded in terms of paths, which is a purely equational notion.

\subsubsection{Observational semantics, operational semantics, and computational adequacy}

As an application of the theory developed here, we define three denotational models of a nondeterministic version of \pcf{} corresponding to the lower, upper, and convex synthetic powerdomains. As in the tradition of denotational semantics, for each model we prove the \emph{computational adequacy} property stating that the denotational and operational semantics agree at the base type. To streamline the proof, we introduce an intermediate semantics called the \emph{observational semantics} that allows us to factor out reasoning that is uniform across all three powerdomains. Intuitively, the idea behind the observational semantics is to ``chain together'' the small-step transitions of the operational semantics using the semantic powerdomain constructions; thus, in a sense, we may think of the observational semantics as being halfway between the denotational and operational semantics. Consequently, the adequacy proof consists of two steps. First, we prove that all three denotational semantics coincide with the corresponding observational semantics by a familiar logical-relations construction, using the \emph{semantic $\top\top$-lifting}~\cite{katsumata:2005} for the computational effect of nondeterminism. Next, we derive the specific correspondence between the observational semantics and the operational semantics for each powerdomain. Here, because the observational semantics is closely related to the operational semantics, the results follow from straightforward fixed-point induction proofs.

\subsection{Related work}

\subsubsection{Synthetic powerdomain constructions}\label{sec:synthetic-powerdomains}
As already mentioned, Phoa and Taylor's work on the synthetic Plotkin/convex powerdomain is the only prior work we are aware of in this setting. Our own work extends their approach to account for the lower and upper powerdomains. In addition, whereas we focus on the role of powerdomains in denotational semantics, the main contribution of Phoa and Taylor~\cite{phoa-taylor:1990} is the characterization of the \emph{observational order} on the synthetic convex powerdomain, where $x \le y$ in the observational order when $f(x) = \top$ implies $f(y) = \top$ for all $f : X \to \I$. They showed that, on the convex powerdomain, the \emph{Egli-Milner order} is contained in the observational order and the converse inclusion holds up to a double negation. Although the observational order need not coincide with the path order in general, we speculate that their method may be adapted to show that the Egli-Milner order is also contained in the path order, assuming that the interval \I{} satisfies enough \emph{locality principles} in the sense of Sterling and Ye~\cite{sterling-ye:2025}; similar results for the lower and upper powerdomains may also be conjectured.

\subsubsection{Computational adequacy in synthetic domain theory}

There are several approaches to proving computational adequacy for \pcf{}-like languages in synthetic domain theory~\cite{simpson:1999,simpson:2004,sterling-harper:2022,matache-moss-staton:2022,niu-sterling-harper:2024}. Our approach is closest in spirit to Simpson's early work~\cite{simpson:1999}, which proves an adequacy theorem for \pcf{} internal to a synthetic domain theory topos. To define the operational semantics, Simpson~\cite{simpson:1999} assumes that the natural numbers object \Nat{} of the ambient topos is a predomain, but his later work~\cite{simpson:2004} shows that this assumption may be eliminated by defining the operational semantics using the \emph{computational natural numbers} (the natural numbers object in the category of predomains, which need not be preserved by the inclusion into the ambient topos).

In this paper we follow the approach of Niu, Sterling, and Harper~\cite{niu-sterling-harper:2024}, who work in terms of the aforementioned observational semantics (called the computational semantics in that work). We prefer this approach because it lets us show that the denotational and observational semantics coincide without assuming \Nat{} is a predomain; moreover, whenever this assumption holds, the observational semantics coincides with the operational semantics. Thus the observational semantics provides a finer-grained perspective on the computational adequacy property.

We also find the approach of defining the operational semantics purely in terms of the (computational) natural numbers somewhat unsatisfying. Traditional programming-language syntax can be encoded using natural numbers, but more expressive languages such as dependent type theories are increasingly defined as \emph{equational theories}; at this level, it seems onerous to justify one's syntax in terms of natural numbers. On the other hand, natural numbers play a key role in interpreting adequacy theorems internal to a synthetic domain theory topos in external mathematics. Indeed, Simpson~\cite{simpson:1999} identifies a purely logical criterion involving propositions of the form $\exists n :\Nat.~p(n)$ in the topos that ensures that the internal adequacy theorem implies an \emph{external} adequacy property (\ie{} for \pcf{} defined in \SET). We do not prove any external adequacy results in this paper, and we leave it to future work to find a similar criterion in our setting that relates internal and external adequacy properties.

\subsubsection{Powerdomains in type theory}

A more distant relative is the work of Møgelberg and Vezzosi~\cite{modelberg-vezzosi:2021} on \emph{guarded} powerdomains in \emph{guarded type theory}~\cite{sgdt:2011}. Guarded type theory is an axiomatization of guarded recursion~\cite{nakano:2000} in which types may be thought of as synthetic \emph{guarded} domains. In general, the approaches of synthetic domain theory and synthetic guarded domain theory are somewhat incomparable, since the former is modeled in domains whereas the latter is modeled in certain complete metric spaces. Practically speaking, guarded domains have the advantage of being closed under all type-theoretic connectives, making them useful for modeling very sophisticated computational effects such as higher-order store~\cite{sterling-gratzer-birkedal:2023}. On the other hand, maps of guarded domains (and guarded recursive domain equations) only possess fixed points up to \emph{weak bisimulation}, which can complicate semantic developments. For instance, the guarded powerdomain of Møgelberg and Vezzosi~\cite{modelberg-vezzosi:2021} for must-convergence is only (conjectured to be) a monad up to weak bisimulation. In contrast, the powerdomain constructions in this paper all carry canonical monadic structures in the ambient equational theory, rendering synthetic powerdomains a more ergonomic metalanguage for nondeterminism.

\subsection{Contributions and outline}

Our contributions may be summarized as follows.

\begin{enumerate}
    \item We extend the method of Phoa and Taylor to construct the lower and upper powerdomains in synthetic domain theory by means of quotient inductive types.
    \item We show that these are \emph{powerdomain constructions} (in a technical sense we discuss in \cref{sec:powerdomains}); in particular this means that they preserve lift algebras, a nontrivial property in synthetic domain theory.
    \item We use the resulting constructions to define denotational and observational semantics for a nondeterministic version of \pcf{} for each of the lower, upper, and convex powerdomains.
    \item We show that these denotational and observational semantics coincide by a logical-relations argument, and that the observational and traditional operational semantics coincide when the interval object \I{} is closed under enough joins, thereby establishing a condition under which computational adequacy holds.
\end{enumerate}

This also provides a rough outline of the paper. In \cref{sec:sdt} we provide the necessary background on synthetic domain theory. In \cref{sec:powerdomains} we define what constitutes a powerdomain construction; we recall the construction of the synthetic convex powerdomain of Phoa and Taylor~\cite{phoa-taylor:1990} and, following an analogous method, define the lower and upper counterparts, showing that all three are indeed powerdomain constructions. In \cref{sec:den-sem} we use the resulting powerdomains to define the denotational semantics of a nondeterministic version of \pcf{} called \pcfnd{}. In \cref{sec:op-sem} we give the operational semantics of \pcfnd{} in terms of a traditional small-step transition semantics and use it to define may and must evaluation. In \cref{sec:ob-sem} we define the \emph{observational semantics} that act as a bridge between the denotational and operational semantics, and we prove the \emph{observational adequacy} theorem stating that denotational and observational semantics coincide. This coincidence is exploited in \cref{sec:adequacy} to show that the denotational and operational semantics coincide whenever \I{} is closed under enough joins. In \cref{sec:conclusion} we describe models of synthetic domain theory in which our results hold and suggest future work.

\section{Synthetic domain theory}\label{sec:sdt}

In the following we work in the constructive mathematics of a topos \ECat{} equipped with a distinguished object \I{}. 

\subsection{Partial maps and observations}

The logical role of \I{} is to classify subsets that occur as the \emph{domain of definition of partial maps}. 

\begin{defn}
    For a stable class $\mathcal{M}$ of monomorphisms closed under identity and composition (referred to as a \emph{dominion} in \cite{rosolini:1986}), an \emph{$\mathcal{M}$-partial map} $X \to Y$ is a span $X \hookleftarrow U \to Y$ such that the domain of definition or termination support $U \hookrightarrow X$ is in $\mathcal{M}$.
\end{defn}

\begin{defn}
    A \emph{dominance} is a collection of propositions $D$ for which there is a dominion $\mathcal{M}$ such that every monomorphism in $\mathcal{M}$ has a characteristic map factoring through $D$. 
\end{defn}

In type-theoretic language a dominance is a collection of propositions $D$ containing $\top$ and closed under dependent sums; for $\phi : D$ and $f : (\phi = \top) \to D$ we write $\phi \mathbin{\angle} f : D$ for the dependent sum $\Sigma_{\phi:D}.~f : D$. 

\begin{ax}
  The object \I{} forms a dominance. 
\end{ax}

\begin{defn}
    Let $D$ be a dominance associated to some dominion $\mathcal{M}$. For any $Y$, a \emph{$D$-partial map classifier} for $Y$ is a map $\eta : Y \hookrightarrow L Y$ such that for any $\mathcal{M}$-partial map $X \hookleftarrow U \to Y$ there is a unique characteristic map $X \to L Y$ such that $U \hookrightarrow X$ is the pullback of $\eta : Y \hookrightarrow L Y$ along $X \to L Y$.
\end{defn}

In dependent type theory we may define the $\I{}$-partial map classifier as $L X \defeq \Sigma_{i : \I}. (i = \top) \to X$; this definition is functorial and carries a monad structure $\LL = (L, \eta_\LL, \mu_\LL)$. 

\begin{ex}\label[ex]{ex:dominance}
  The subobject classifier $\Omega$ and the type $2 = \MkSet{\bot,\top}$ are dominances; the associated dominion consists of all subsets for the former and decidable subsets for the latter. 
\end{ex}

Logically we view \I{} as the \emph{classifier of observations} --- every subset $S \subseteq X$ whose membership is \emph{observable} has a characteristic map factoring through \I{}. In terms of topology this means that we may view maps $X \to \I{}$ as \emph{open subsets} $S \subseteq_\I X$. Accordingly these maps induce a preorder usually referred to as the \emph{specialization preorder} on $X$. 

\begin{defn}\label{def:observational-preorder}
    Given $x,y: X$, we have that $x \le y$ on the \emph{observational preorder} if and only if $(f x = \top) \to (f y = \top)$ for all maps $f : X \to \I{}$.
\end{defn}

By way of the specialization topology on the observational preorder, the choice of \I{} determines a \emph{global} topological structure on types for which all maps are continuous. The goal is to ensure that there are enough \I{}-partial maps for the purposes of domain theory but not so many that they lack observational interpretation. Observe that neither of the dominances from \cref{ex:dominance} fulfills these goals. If we pick $\I{} = \Omega$ to consist of all propositions then the observational preorder becomes distinct from the implication order; in particular the negation map $\Omega \to \Omega$ is not monotone with respect to the implication order and moreover cannot possess a fixed point. A similar problem happens when we set $\I{} = 2$; in addition, the class of $2$-partial maps $X \to L Y$ actually corresponds to total maps $X \to Y + 1$ and does not admit all maps needed for domain theory and denotational semantics. 

The art of synthetic domain theory lies in picking a dominance $2 \subseteq \I \subseteq \Omega$ so that we can represent enough maps without admitting too many. In some models of synthetic domain theory (such as those based on realizability~\cite{rosolini:1986}) one may give an internal description of a dominance \I{} satisfying the aforementioned goals (for instance as the class of \emph{semidecidable} propositions). In contrast, the dominances in sheaf models~\cite{fiore-rosolini:1997,fiore-plotkin:1996,sterling-harper:2022} are not known to have such internal presentations. In this paper we axiomatize the dominance \I{} so that the results obtained may be instantiated in both kinds of models.

\subsection{Interval and paths}

Dual to the logical role of \I{} as the classifier of observations is its geometric role as an \emph{interval object}. 

\begin{ax}
    The dominance \I{} is closed under $\bot$.  
\end{ax}

Under this axiom we may view $\I$ as an interval whose boundary is $0 \defeq \bot$ and $1 \defeq \top$.

\begin{defn}
  A \emph{path} $\alpha : x \leadsto y$ in a type $X$ is a map $\alpha : \I \to X$ such that $\alpha(0) = x$ and $\alpha(1) = y$; we write $\partial \alpha = (\alpha~ 0, \alpha~ 1)$ for the \emph{boundary} of a path and $(x \leadsto y) \mathbin{@} i$ for the function application $\alpha~i$.
\end{defn}

\begin{defn}
    A type is \emph{boundary separated} when there is at most one path with a given boundary. 
\end{defn}

In contrast to the observational preorder, paths $x \leadsto y \leadsto z$ need not compose or even form a relation in general, \ie{} be boundary separated. Nonetheless in this paper it is actually paths that play the critical role of evincing the order-theoretic structure of \emph{synthetic domains} (to be defined in \cref{sec:predomains}), which are distinguished types on which the path relation is well-behaved and share properties familiar from analytic domain theory. Although the observational preorder and path relation need not coincide for general types, they both coincide with the implication order on \I{} when we assume \emph{Phoa's principle}~\cite{phoa:1991}. 

\begin{defn}
    The interval \I{} satisfies \emph{Phoa's principle} when every pair $(\phi,\psi) : \I^2$ with $(\phi = \top) \to (\psi = \top)$ corresponds to a unique path $\I \to \I$ whose boundary is $(\phi,\psi)$.
\end{defn}

\begin{ax}
    The interval \I{} satisfies Phoa's principle.
\end{ax}

An important consequence of Phoa's principle is that all three relations on the interval coincide. 

\begin{prop}
    The observational preorder, implication order, and path relation on \I{} all coincide. 
\end{prop}

\subsection{Order structure, localizations, and reflective subcategories}

As a first approximation, a synthetic predomain is a type that is cocomplete with respect to joins of suitable shapes. This is analogous to the closure under $\Nat$-joins in the case of \ocpo{}s. Moreover, since the universal characterizations of the upper and lower powerdomains are order-theoretic~\cite{hennessy-plotkin:1979}, we want to ensure that the path relation on a synthetic predomain is actually a partial order. We may isolate such a class by means of a \emph{localization} --- a way to identify types that ``believe'' certain maps are isomorphisms. 

\begin{defn}\label[defn]{def:local}
    A type $X$ is \emph{(internally) local} with respect to a map $f : A \to B$ when the induced precomposition map $X^f : X^B \to X^A$ is an isomorphism; dually such an $f$ is called an \emph{$X$-isomorphism}. 
\end{defn}

Observe that $X$ is local with respect to $f : A \to B$ when every map $g : A \to X$ has a \emph{unique extension} $\overline{g} : B \to X$, where an extension is a map $h : B \to X$ such that $hf = g$. For a given collection $\mathcal{F}$ of maps we say $X$ is $\mathcal{F}$-local when it is local with respect to every map in $\mathcal{F}$. In a Cartesian closed category $\CCat$ we write $\CCat_\mathcal{F}$ for the full subcategory of (internally) $\mathcal{F}$-local objects. The central axiom of synthetic domain theory involves locality with respect to the canonical inclusion map $\omega \hookrightarrow \overline{\omega}$, where $\omega$ is the initial lift algebra and $\overline{\omega}$ is the final lift coalgebra.

\begin{defn}
    A type is \emph{complete} when it is local with respect to $\omega \hookrightarrow \overline{\omega}$. 
\end{defn}

Here it may be helpful to note the formal analogy between the completeness property and the fact that in the category of \ocpo{}s every object is local with respect to the inclusion $\MkSet{0 \le 1 \dots} \hookrightarrow \MkSet{0 \le 1 \dots \le \infty}$. 

\begin{prop}\label[prop]{prop:fixed-points}
    Let $X$ be a complete type that is equipped with an algebra for the lift monad \LL{}; then every map $X \to X$ possesses a fixed point.
\end{prop}

In general the path relation need not be well-behaved on complete types; we refine the notion of a synthetic (pre)domain by considering additional localizations. 

\begin{prop}[\cite{niu-sterling-harper:2024}]\label[prop]{prop:partial}
    There is a finite collection of maps $\mathcal{F}$ such that the path relation on $\mathcal{F}$-local types is a partial order.  
\end{prop}

The significance of isolating the desirable properties of our candidate notion of synthetic predomain in terms of locality is that localization at a small collection of maps $\mathcal{F}$ obtains a category closed under many constructions used to interpret basic constructions in domain theory. 

\begin{prop}[\cite{rijke-etal:2020}]
    The full subcategory $\ECat_\mathcal{F}$ of $\mathcal{F}$-local types for a small collection of maps $\mathcal{F}$ is an internally complete and cocomplete reflective subuniverse in the ambient topos $\ECat$.
\end{prop}

The reflector $\mathcal{S} : \ECat \to \ECat_\mathcal{F}$ induces an idempotent monad on $\ECat$ with unit $\eta_\synRef$.  Reflectivity implies that $\ECat_\mathcal{F}$ is \emph{closed} under the ambient limits of the topos $\ECat$; consequently we may use the ordinary internal language to define and manipulate limits in $\ECat_\mathcal{F}$. In contrast the inclusion $\ECat_\mathcal{F} \hookrightarrow \ECat$ need not preserve colimits of $\mathcal{F}$-local objects, which means we need to be careful in applying the universal property of colimits in $\ECat_\mathcal{F}$.

\begin{prop}[$\mathcal{S}$-induction]\label[prop]{prop:S-induction}
    To prove a property $\Phi$ holds for all $u : \mathcal{S}X$, it suffices to show that $\Phi$ is $\mathcal{F}$-local and that $\Phi(\eta_\synRef(x))$ for every $x : X$. 
\end{prop}

\begin{proof}
    From the assumptions that $\Phi$ is $\mathcal{F}$-local and $\Phi(\eta_\synRef(x))$ holds for every $x : X$ we obtain the \synRef{}-reflection $X \to \Phi$, where the unique extension of a map $f : X \to Y$ into an $\mathcal{F}$-local type $Y$ is given by $\Phi \hookrightarrow \synRef{}X \xrightarrow{\overline{f}} Y$, where $\overline{f}$ is the unique extension obtained by the universal property of $\synRef{}X$. Observe that $\Phi \hookrightarrow \synRef{}X$ is an isomorphism whose inverse $\synRef{}X \to \Phi$ is the unique extension of $X \to \Phi$ along $\eta : X \to \synRef{}X$; that these maps are mutually inverse follow from the universal property of \synRef{}-reflections. 
\end{proof}

\subsection{Lifting structure \vs{} pointedness}

A domain is a predomain equipped with an algebra for the lift monad \LL{}. In analytic domain theory a domain is equivalently a predomain equipped with a least element. In synthetic domain theory, it is unclear whether the existence of a least element is enough to reconstruct the lifting structure since the well-known construction \cite{kock:1991} of the lifting structure uses the fact that predomains are \emph{weakly cocomplete} (closed under joins of subterminal subsets). Since it is unclear how to define the corresponding figure shape in synthetic domain theory, it is unknown whether lift algebras and pointed types are equivalent in general. 

This distinction becomes relevant when we consider power\emph{domain} constructions. The three well-known constructions of upper, lower, and convex powerdomains all preserve the least element, which means they send domains to domains in analytic domain theory. On the other hand, as we discussed above, this is not enough to ensure that the analogous constructions preserve lifting structure in synthetic domain theory, a problem foresaw by Phoa~\cite{phoa:1994}. 

\subsection{The \sierp{} cone}

To show that the various synthetic powerdomain constructions preserve lifting structure we need to define maps $L(\mathcal{P}X) \to \mathcal{P}X$ satisfying the monad algebra laws. Because the partial map classifier $L$ has a \emph{right-handed} universal property in general it is unclear how one may define the required algebra maps \emph{out of} the lift of a type.   

We may resolve this tension by restricting to types $X$ for which $L X$ has an additional \emph{left-handed} universal property. In analytic domain theory this is true for all types: to define a continuous map $L X \to Y$ it suffices to give $f : X \to Y$ and an element $y : Y$ such that $y$ is below the image of $f$. In synthetic domain theory $L X$ has this universal property when it is the \emph{\sierp{} cone}~\cite{pugh-sterling:2025} of a type $X$, which may be characterized as the cospan $\MkSet{\bot} \hookrightarrow X_\bot \hookleftarrow X$ fitting into the pushout diagram depicted below.

\[
\begin{tikzpicture}
    \SpliceDiagramSquare<sq/>{
        width = 2cm,
        nw = X,
        sw = \MkSet{\bot},
        ne = \I{} \times X,
        se = X_\bot,
        south/style = embedding,
        north = 0 \times X,
        se/style = pushout,
    }
    \node[right = of sq/ne, xshift=1cm] (X) {$X$}; 
    \draw[->](X) to node[above] {$1 \times X$} (sq/ne); 
    \draw[embedding](X) to node {} (sq/se); 
\end{tikzpicture}
    \]

Whenever the path order on $X$ is a partial order we may think of the \sierp{} cone of a type as freely adjoining to it a least element. To define a map $X_\bot \to Y$ out of $X_\bot$ it suffices to provide a distinguished element $y : Y$ and a map $f : \I{} \times X \to Y$ such that the restriction $f(0, -) : X \to Y$ is constantly $y$, mirroring the aforementioned analytic monotonicity requirement. Pugh and Sterling~\cite{pugh-sterling:2025} have identified when the \sierp{} cone $X_\bot$ on a type $X$ is computed as the partial map classifier $L X$. Notably \opcit{} shows that there is a small localization class for which the two constructions coincide.

\begin{prop}[\cite{pugh-sterling:2025}]\label[prop]{prop:sierp}
    Assuming Phoa's principle, there exists a small collection of maps $\mathcal{F}$ such that in the reflective subcategory of $\mathcal{F}$-local types the cospan $\MkSet{(\bot, \kw{abort})} \xhookrightarrow{} L X \xhookleftarrow{\eta_X} X$ is the \sierp{} cone on $X$.
\end{prop}

\subsection{Synthetic predomains}\label{sec:predomains}

We are now in a position to identify classes of predomains to which our results apply. 

\begin{defn}
    For any type $S$, a type $X$ is $S$-\emph{replete} when it is $f$-local whenever $f$ is an $S$-isomorphism. Unfolding definitions (\cref{def:local}) this means that $X$ is $S$-replete when $X^f : X^B \to X^A$ is an isomorphism whenever $S^f : S^B \to S^A$ is an isomorphism for all maps $f : A \to B$. We say a type is simply \emph{replete} when it is \I{}-replete.
\end{defn}

Although the localization class for replete types is not indexed by a small type, it is well known that the category of replete types is a reflective exponential ideal and moreover the smallest such containing \I{}~\cite{hyland:1991}. However nothing about our development depends on repleteness \emph{per se}; for the sake of generality we simply assume a category of predomains $\CCat$ with suitable properties.

\begin{ax}\label[ax]{ax:sdt}
    There is a reflective subuniverse $\CCat \hookrightarrow \ECat$ containing \I{} satisfying the following.
    \begin{enumerate}
        \item Every type $X$ in \CCat{} is complete and the path space $X^\I$ determines a partial order $\pathle$ on $X$.
        \item The subuniverse \CCat{} is closed under lifting.
        \item In \CCat{} the \sierp{} cone is computed as the partial map classifier.
    \end{enumerate}
\end{ax}

For instance we may define $\CCat$ by localizing at the maps given by \cref{prop:partial,prop:sierp} and $\MkSet{\omega \cap S \hookrightarrow S \mid S \subseteq_\I{} \overline{\omega}}$, where the latter collection of maps ensures closure under lifting~\cite{fiore-rosolini:1997:cpos,simpson:1999}.

\begin{defn}
    A \emph{predomain} is a type in $\CCat$, and a \emph{domain} is a predomain $\mathbb{L}$-algebra. 
\end{defn}

On predomains we recover the classical equivalence of lift algebras and pointed types on predomains.

\begin{prop}
    Every \LL{}-algebra $\alpha : L X \to X$ on a predomain $X$ evinces a least element in $X$ given by $\alpha(\bot, \kw{abort})$. 
\end{prop}

\begin{proof}
    By the above proposition we have a path $(\bot, \kw{abort}) \leadsto \eta x$ for all $x : X$, so the result follows by postcomposing this path with $\alpha$ since we have $\alpha(\eta(x)) = x$. 
\end{proof}

Because lifting computes the \sierp{} cone in the category of predomains, equality of maps defined on partial elements may be checked on just the totally defined and undefined elements~\cite[Corollary 4]{sterling:2024}.

\begin{prop}\label[prop]{prop:epi}
    The induced map $1 + X \to L X$ is an epimorphism in the category of predomains. 
\end{prop}

\begin{restatable}{prop}{IsDomain}\label[prop]{prop:is-domain}
    A predomain is a domain if and only if it possesses a least element on the path order. 
\end{restatable}

\begin{proof}
    Let $X$ be a predomain with least element $\bot$, \ie{} there exists a necessarily unique path $\bot \leadsto x$ for all $x : X$. We construct the candidate algebra map $\alpha : L X \to X$ by the universal property of the \sierp{} cone. 
    
    \begin{enumerate}
        \item The distinguished point is $\bot$.
        \item The monotonicity map $\alpha^\zeta : \I \times X \to X$ sends $(i, x)$ to $(\bot \leadsto x) \mathbin{@} i$. 
    \end{enumerate}

    First we check the unit law. Observe that by definition of the \sierp{} cone the composition $X \xhookrightarrow{\eta} L X \xrightarrow{\alpha} X$ is the restriction $\alpha^\zeta(1, -) : X \to X$, which is $x$ by construction. For the multiplication law, we first give a left-handed description of the lift monad multiplication. Observe that by \cref{prop:bottom} we have that $(\bot, \kw{abort}) : L X$ is the least element of $L X$, so instantiating the construction above we have a map $\mu' : L^2 X \to L X$. We claim that $\mu' = \mu$ is the multiplication of the lift monad. By \cref{prop:epi} it suffices to show that $\MkSet{(\bot, \kw{abort})} \xhookrightarrow{} L^2 X \xrightarrow{\mu, \mu'} L X$ and $L X \xhookrightarrow{\eta} L^2 X \xrightarrow{\mu, \mu'} L X$ commute; the former follows immediately by computation, and the latter follows because $\mu\eta = 1$ by the unit law and $\mu'\eta = 1$ by the same argument above. 

    Thus it suffices to check the multiplication law with respect to $\mu'$; again by \cref{prop:epi} it suffices to check that the following diagrams commute. 

    \[
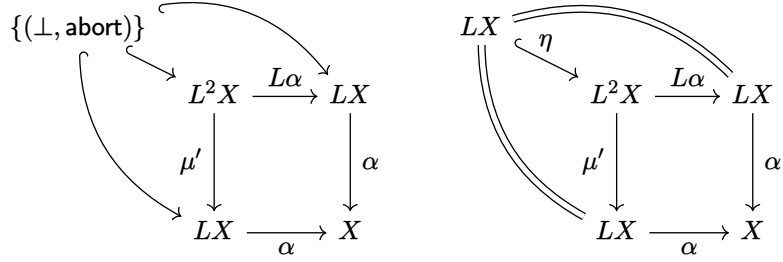
    \begin{array}{c@{\qquad}c}
        \begin{tikzpicture}[diagram, scale=0.9, transform shape]
            \SpliceDiagramSquare<sq/>{
                nw = L^2 X,
                sw = L X,
                ne = L X,
                se = X,
                west = \mu',
                south = \alpha,
                east = \alpha,
                north = L \alpha,
            }
            \node[left = of sq/nw, yshift=1cm] (X) {$\MkSet{(\bot, \kw{abort})}$};
            \draw[embedding](X) to (sq/nw);
            \draw[embedding](X) to[bend right] (sq/sw);
            \draw[embedding*](X) to[bend left] (sq/ne);
        \end{tikzpicture}
        &
        \begin{tikzpicture}[diagram, scale=0.9, transform shape]
            \SpliceDiagramSquare<sq/>{
                nw = L^2 X,
                sw = L X,
                ne = L X,
                se = X,
                west = \mu',
                south = \alpha,
                east = \alpha,
                north = L \alpha,
            }
            \node[left = of sq/nw, yshift=1cm] (X) {$L X$};
            \draw[embedding](X) to node[above] {$\eta$} (sq/nw);
            \draw[=,double distance=2pt](X) to[bend right] (sq/sw);
            \draw[=,double distance=2pt](X) to[bend left] (sq/ne);
        \end{tikzpicture}
    \end{array}
    \]

    For the left diagram, we have that the outer triangles commute by direct computation; for the right diagram, we have that both triangles are the identity by the unit law of monads and the unit algebra law, respectively. 
\end{proof}

\subsection{Domains and recursion}

\begin{restatable}{prop}{BottomLeast}\label[prop]{prop:bottom}
    For a predomain $X$ the least element of $L X$ is given by the unique divergent element.
\end{restatable}

Synthetic domain theory also supports a version of fixed-point induction for suitable propositions. 

\begin{defn}
    A subset of a domain is \emph{admissible} when it is complete and closed under the least element. 
\end{defn}

\begin{prop}\label[prop]{prop:complete}
    The category of complete types is closed under internal limits. 
\end{prop}

\begin{cor}\label[cor]{cor:intersection}
    Admissible types are closed under arbitrary intersections. 
\end{cor}

\begin{prop}\label[prop]{prop:admissible}
    Let $\phi : X \to \I$ be an \I{}-subset of a domain $X$ such that $\phi(\bot_X) = \bot$ and $\psi$ an \I{}-proposition. Then the subset $\MkSet{x : X \mid \phi~x \to \psi}$ is admissible. 
\end{prop}

\begin{proof}
    The subset $\MkSet{x : X \mid \phi~x \to \psi}$ contains the least element of $X$ as $\phi(\bot_X) = \bot$, and it is complete because it may be expressed as the following pullback diagram of complete types. 

    \[
        \DiagramSquare{
            width = 3cm,
            nw = P, 
            sw = X,
            ne = \I^\I,
            se = \I^2,
            east = \partial,
            south = {\langle \phi, \psi \circ ! \rangle},
            nw/style = pullback,
        }
    \]
\end{proof}

\begin{prop}
    Given an admissible subset $\Phi \subseteq X$ and $f : X \to X$, the fixed point $\mu f : X$ of $f$ is contained in $\Phi$ whenever $f~\Phi \subseteq \Phi$ (\ie{} $\Phi$ is a pre-fixed point of $f$).  
\end{prop}

\subsection{Basic predomains}

Our axioms (\cref{ax:sdt}) imply that the category of predomains contains \emph{computational} analogues to base types such as the booleans and natural numbers. Following Simpson~\cite{simpson:2004} we refer to these objects as computational because in general they need not be preserved by the inclusion $\CCat \to \ECat$ unless the dominance \I{} is \emph{closed} under joins of a corresponding shape; for instance the type 2 of booleans is a predomain whenever \I{} has binary joins that are preserved by $(- = \top) : \I{} \hookrightarrow \Omega$~\cite{reus:1995}. For the purposes of denotational semantics we would like to be able to \emph{observe} when a natural number is a possible outcome of a nondeterministic program; moreover for the purposes of proving computational adequacy we further require that the computational natural numbers have decidable equality. By the preceding remark we may ensure this by closing the dominance under finite joins.

\begin{ax}\label[ax]{ax:finite-join}
    The dominance \I{} has finite joins that are preserved by $(- = \top) : \I{} \hookrightarrow \Omega$.
\end{ax}

\begin{restatable}{prop}{NatSigmaEquality}\label[prop]{prop:nat}
    The predomain \NatP{} of natural numbers has decidable equality. 
\end{restatable}

\begin{proof}
    First observe that because the ambient natural numbers is decidable, the characteristic map for the diagonal set $\Delta_\Nat \hookrightarrow \Nat^2$ factors through the booleans 2, which is a predomain by virtue of \cref{ax:finite-join}. Thus we have a unique extension as follows.

    \[\Local[\chi][->]{\Nat^2}{2}{\synRef\Nat^2 \cong \NatP^2}[][\overline{\chi}][3cm]\]

    We aim to show that $\overline{\chi}(x, y) = \top$ implies $x = y$ for all $x, y : \NatP$. We proceed by \synRef{}-induction with the subset $\Phi = \MkSet{(x, y) \mid \overline{\chi}(x,y) \to x = y}$. 

    \begin{enumerate}
        \item First we need to show that $\Phi(\eta_\synRef(m), \synRef(n))$ holds for all $m,n : \Nat$. But this holds since $\chi$ is the characteristic map of $\Delta_\Nat \hookrightarrow \Nat^2$.
        \item Next we need to show that $\Phi$ is a predomain, which follows from the fact it may be expressed as the following pullback diagram. 

        \[
        \DiagramSquare{
            nw = \Phi,
            sw = \NatP^2,
            ne = L(\Delta_{\NatP}), 
            se = L(\NatP^2),
            nw/style = pullback,
        }
        \]

        In the above the right map is given by the lift of the inclusion $\Delta_{\NatP} \hookrightarrow \NatP^2$ and the bottom map is defined by sending $x, y : \NatP$ to $(\overline{\chi}(x,y), x, y)$.
    \end{enumerate}
\end{proof}

\section{Synthetic powerdomains}\label{sec:powerdomains}

The universal characterization of powerdomains as certain \emph{continuous algebras}~\cite{hennessy-plotkin:1979,gunter-etal:1989} serves as a blueprint for constructing powerdomains in synthetic domain theory. Indeed Phoa and Taylor~\cite{phoa-taylor:1990} had already carried out the case for the Plotkin/convex powerdomain; here we recall their construction. By the characterization of the Plotkin powerdomain as the \emph{free} semilattice in a given category of predomains, it suffices to find a left adjoint $\CCat \to \kw{SL}(\CCat)$ for the top functor in the following diagram of forgetful functors, where we write $\kw{SL}(\CCat)$ for the category of semilattices and semilattice homomorphisms in $\CCat$.

\[
\DiagramSquare{
    nw = \kw{SL}(\CCat),
    sw = \kw{SL}(\ECat),
    ne = \CCat,
    se = \ECat,
    north/style = embedding,
    south/style = embedding,
    east/style = embedding,
    west/style = embedding,
}
\]

Observe that all other subcategory inclusions have left adjoints. By \cref{ax:sdt} the full subcategory $\CCat \to \ECat$ is right adjoint to the synthetic predomain reflection $\mathcal{S} : \ECat \to \CCat$, and the latter also restricts to the left adjoint to the inclusion $\kw{SL}(\CCat) \hookrightarrow \kw{SL}(\ECat)$\footnote{This is because the synthetic predomain reflection preserves finite products and thus preserves models of algebraic theories (such as semilattices) and their homomorphisms.}; the left adjoint to the inclusion $\kw{SL}(\ECat) \hookrightarrow \ECat$ is given by the free semilattice construction $K : \ECat \to \kw{SL}(\ECat)$, which is available in any topos. Thus we may define a functor $\CCat \to \kw{SL}(\CCat)$ by the composition $\CCat \hookrightarrow \ECat \xrightarrow{K} \kw{SL}(\ECat) \xrightarrow{\mathcal{S}} \kw{SL}(\CCat)$; it is left adjoint to $\kw{SL}(\CCat) \hookrightarrow \CCat$ because the subcategory inclusion $\CCat \to \ECat$ is full and faithful. In the following we fix a general notion of a powerdomain construction (which includes the construction of \opcit{}). 

\begin{defn}\label[defn]{def:powerdomain}
    A \emph{powerdomain construction} is a functor $\mathcal{P} : \CCat \to \kw{SL}(\CCat)$ that sends domains to domains and such that postcomposing with the forgetful functor $\kw{SL}(\CCat) \to \CCat$ is a monad. 
\end{defn}

Unfolding the above, a powerdomain on a predomain $X$ is a predomain semilattice algebra $(\mathcal{P}X, \vee_\mathcal{P})$ equipped with a singleton operation $\singleton{-}_{\mathcal{P}} : X \to \mathcal{P}X$ and monadic lift operation $\Tilde{-}_{\mathcal{P}} : (X \to \mathcal{P}Y) \to (\mathcal{P}X \to \mathcal{P}Y)$. We may omit the subscripts $\mathcal{P}$ when there is no ambiguity. 

\begin{rem}
    The semilattice operation $\vee$ in a powerdomain construction has no \emph{a priori} relation to the path ordering; it is however by definition always the join for the \emph{subset ordering} $S \subseteq T \defeq (S \vee T = T)$. 
\end{rem}

\subsection{Distributive laws and monad structure}

Despite the name a powerdomain $\mathcal{P}X$ need not be pointed unless $X$ is pointed. In general we need to precompose with lifting to obtain a domain suitable for interpreting recursive functions. For the purposes of denotational semantics we want to show that this composition is a monad.

\begin{prop}\label[prop]{prop:monad}
    Let $\MM = (M, \eta_\MM, \mu_M)$ be a monad on \CCat{} such that $\eta_\MM$ preserves the least element. Then there is a distributive law $\tau : \LL\MM \to \MM\LL$ of the lift monad over $\MM$. 
\end{prop}

\begin{proof}
    A distributive law $\tau : \LL\MM \to \MM\LL$ is equivalent to a lifting of the monad $\MM$ to the category of lift algebras on $\CCat$. By \cref{prop:is-domain} this is equivalent to finding a lifting to the category of pointed predomains and strict (least-element preserving) maps; from the assumption that the unit $\eta : X \to M X$ preserves least elements it is routine to verify that $M$ restricts to a functor on strict maps and the multiplication of $\MM$ is strict. 
\end{proof}

\subsection{The convex powerdomain}

Phoa and Taylor's construction of the convex powerdomain $\CCat \to \kw{SL}(\CCat)$ we outlined above is nearly shown to be a powerdomain construction in the sense of \cref{def:powerdomain}; it just remains to show that the construction preserves domains. First we observe that the free semilattice $K X$ on a type $X$ may be defined in terms of a \emph{quotient inductive type}~\cite{frumin-etal:2018}, which endows the type $K X$ with the following induction principle. 

\begin{prop}\label[prop]{prop:K-induction}
    To prove a proposition $\Phi$ holds for all $S : K X$, it suffices to show that $\Phi$ holds for all singletons $\MkSet{x}$ and is closed under the semilattice operation. 
\end{prop}

To show that $\kw{C} \defeq \CCat \hookrightarrow \ECat \xrightarrow{K} \kw{SL}(\ECat) \xrightarrow{\mathcal{S}} \kw{SL}(\CCat)$ sends domains to domains, we show that each component functor preserves the property of \emph{pointedness} using the respective induction principles. 

\begin{defn}\label[defn]{def:pointed}
    A type $X$ is \emph{pointed} when there is a distinguished element $x_0 : X$ such that for every $x : X$ there is a path $p_x : x_0 \leadsto x$. 
\end{defn}

\begin{prop}
    If $X$ is pointed, then so is $K X$.
\end{prop}

\begin{restatable}{prop}{PropReflectionPointed}
    If $X$ is pointed, then so is its synthetic reflection $\mathcal{S}X$. 
\end{restatable}

\begin{proof}
    Let $x_0 : X$ be the distinguished point. It suffices to show that $\eta(x_0)$ is the distinguished point in $\mathcal{S}X$. By the principle of $\mathcal{S}$-induction (\cref{prop:S-induction}), it suffices to show that there is a path $q_x : \eta(x_0) \leadsto \eta(x)$ for every $\eta(x)$ and that the subset $\Phi = \MkSet{ u : \mathcal{S}X \mid \exists \alpha : \eta(x_0) \leadsto u}$ is a predomain. For the former, we may define $q_x$ by postcomposing the path $p_x : x_0 \leadsto x$ (which exists by assumption) with $\eta$. For the latter, since the path relation on the synthetic predomain reflection $\mathcal{S}X$ is a partial order, $\Phi$ is equivalent to the subset $\MkSet{ u : \mathcal{S}X \mid \eta(x_0) \pathle u}$. The result then follows because predomains are internally complete and we may exhibit $\Phi$ as the following pullback. 
    \[
    \DiagramSquare{
        width = 3cm,
        nw = \Phi,
        sw = \mathcal{S}X,
        ne = {\mathcal{S}X^\I},
        se = {\mathcal{S}X^2}, 
        east = \partial,
        south = {\langle \eta(x_0), 1 \rangle},
        nw/style = pullback,
    }
    \]
\end{proof}

\begin{cor}
    If $X$ is a domain with least element $\bot : X$, then  $\eta_\synRef\MkSet{\bot}$ is the least element in $\kw{C}X$. 
\end{cor}

By virtue of $\kw{C} : \CCat \to \kw{SL}(\CCat)$ being a left adjoint, it is a powerdomain construction in the sense of \cref{def:powerdomain}, and we may define the convex powerdomain on $X$ to be $\kw{C}X$. 

\subsection{The lower and upper powerdomains}

In contrast to the convex powerdomain, the lower and upper powerdomains are characterized by \emph{in}equational theories, respectively extending the theory of the convex powerdomain by the inequalities $S \sqsubseteq S \vee T$ and $S \vee T \sqsubseteq S$. This suggests that we should carry out the construction outlined in the previous section in the setting of \emph{path order-enriched categories}. In this paper we follow a different approach that is perhaps conceptually simpler. The key observation is that although the path relation need not be a partial order for a general type $X$, it \emph{becomes one} on the synthetic reflection $\mathcal{S}X$. Therefore it suffices to axiomatize the notion of a ``pre-lower semilattice'' and show that such an algebra $X$ has a synthetic predomain reflection $\mathcal{S}X$ that is a genuine lower semilattice, \ie{} a semilattice satisfying $S \pathle S \vee T$. 

\begin{defn}\label[defn]{def:prelowersemilattce}
    A \emph{pre-lower semilattice} is a boundary-separated semilattice $(X, \vee)$ equipped with a function $l : X \times X \times \I{} \to X$ such that $l(S, T, 0) = S$ and $l(S, T, 1) = S \vee T$; a homomorphism of pre-lower semilattices is just a homomorphism of the underlying semilattices.
\end{defn}

\begin{defn}\label[defn]{def:lowersemilattice}
    A \emph{lower semilattice} is a predomain semilattice $(X, \vee)$ such that $S \pathle S \vee T$ for all $S, T : X$; a homomorphism of lower semilattices is just a homomorphism of the underlying semilattices.
\end{defn}

In other words a pre-lower semilattice is just a semilattice $(X, \vee)$ in which there is a path $S \leadsto S \vee T$ (necessarily unique by the requirement that $X$ is boundary separated) for all $S, T : X$ and a lower semilattice is a pre-lower semilattice on a predomain. Here we stipulate boundary separation to ensure that we may construct the free pre-lower semilattice by means of freely adding the required path constructors. Since \cref{def:prelowersemilattce} also makes sense in the category of predomains, we have a category $\kw{PLSL}(\CCat)$ which is moreover equal to the category $\kw{LSL}$ of lower semilattices, inducing a diagram of subcategory inclusions analogous to the case of ordinary semilattices. 

    \[
\DiagramSquare{
    width = 2.5cm,
    nw = {\kw{PLSL}(\CCat) = \kw{LSL}},
    sw = \kw{PLSL}(\ECat),
    ne = \CCat,
    se = \ECat,
    north/style = embedding,
    south/style = embedding,
    east/style = embedding,
    west/style = embedding,
}
\]

Thus to carry out a similar construction to the one in the previous section, it suffices to show that 1) there is a left adjoint to $\kw{PLSL}(\ECat) \hookrightarrow \ECat$ and 2) the synthetic predomain reflection $\mathcal{S} : \ECat \to \CCat$ restricts to the left adjoint to the forgetful functor $\kw{PLSL}(\CCat) \hookrightarrow \kw{PLSL}(\ECat)$. 

\begin{figure}
    \centering

    \begin{multicols}{2}
         \iblock{
        \mhang{\textbf{data}~H X~\textbf{where}}{
            \mrow{\eta : X \to H X}
            \mrow{\vee : H X \to H X \to H X}
            \mrow{\kw{idem} : \Pi_{S : H X}.~S \vee S = S}
            \mrow{\kw{comm} : \Pi_{S, T : H X}.~S \vee T = T \vee S}
            \mrow{\kw{assoc} : \Pi_{S, T, U : H X}.~(S \vee T) \vee U = S \vee (T \vee U)}
            \mrow{l : H X \times H X \times \I \to H X}
            \mrow{l_0 : \Pi_{S, T : H X}.~l(S, T, 0) = S}
            \mrow{l_1 : \Pi_{S, T : H X}.~l(S, T, 1) = S \vee T}
            \mrow{s : \Pi_{\alpha, \beta : \I{} \to H X}.~(\partial \alpha = \partial \beta) \to \alpha = \beta}
        }
        \columnbreak
        \mhang{\textbf{data}~G X~\textbf{where}}{
            \mrow{\eta : X \to G X}
            \mrow{\vee : G X \to G X \to G X}
            \mrow{\kw{idem} : \Pi_{S : G X}.~S \vee S = S}
            \mrow{\kw{comm} : \Pi_{S, T : G X}.~S \vee T = T \vee S}
            \mrow{\kw{assoc} : \Pi_{S, T, U : G X}.~(S \vee T) \vee U = S \vee (T \vee U)}
            \mrow{l : G X \times G X \times \I \to G X}
            \mrow{l_0 : \Pi_{S, T : G X}.~l(S, T, 0) = S \vee T}
            \mrow{l_1 : \Pi_{S, T : G X}.~l(S, T, 1) = S}
            \mrow{s : \Pi_{\alpha, \beta : \I{} \to G X}.~(\partial \alpha = \partial \beta) \to \alpha = \beta}
        }
     }  
    \end{multicols}
    \caption{Left: the free pre-lower semilattice $H X$ as a quotient type; right: the free pre-upper semilattice $G X$. The only difference between the two types is the swapped boundary of the constructor $l$ corresponding to the axioms $S \pathle S \vee T$ and $S \vee T \le S$ respectively. }
    \label{fig:lower-and-upper}
\end{figure}

\begin{restatable}{prop}{PropFreeLowerSemilattice}
    The free pre-lower semilattice on a type $X$ may be defined as the quotient inductive type displayed on the left in \cref{fig:lower-and-upper}.
\end{restatable}

\begin{proof}
    By construction we have that $H X$ is a pre-lower semilattice $(H X, \vee, l)$ and boundary separated. Let $(Y, \vee_Y, l_Y)$ be a pre-lower semilattice and $f : X \to Y$ any map. We aim to construct a unique semilattice homomorphism $\overline{f} : H X \to Y$ such that $(X \xrightarrow{f} Y) = (X \xrightarrow{\eta} H X \xrightarrow{\overline{f}} Y)$. For ordinary constructors we send $\eta(x)$ to $f(x)$, $S \vee T$ to $\overline{f}~S \vee_Y \overline{f}~T$, and $l(S, T, i)$ to $l_Y(\overline{f}~S, \overline{f}~T, i)$. The semilattice laws follow because $Y$ is a semilattice, and similarly $i \mapsto l_Y(\overline{f}~S, \overline{f}~T, i)$ is a path $\overline{f}~S \leadsto \overline{f}~S \vee_Y \overline{f}~T$ by assumption. Lastly for the constructor $s$ it suffices to observe that we stipulate pre-lower semilattices to be boundary separated; similarly this extension is unique because any other homomorphism $f' : H X \to Y$ such that $f' \eta = f$ must be equal to $\overline{f}$ on $S \vee T$ and map $l(S, T, -) : S \leadsto S \vee T$ to a path with the same boundary as $\overline{f}(l(S, T, -))$, so it follows these paths are equal by boundary separation. 
\end{proof}

The second property we need may be shown by $\synRef{}$-induction.  

\begin{restatable}{prop}{PropLeftAdjoint}
    The synthetic predomain reflection $\mathcal{S} : \ECat \to \CCat$ restricts to the left adjoint to the forgetful functor $\kw{PLSL}(\CCat) \hookrightarrow \kw{PLSL}(\ECat)$. 
\end{restatable}

\begin{proof}
    It suffices to show that every pre-lower semilattice $(X, \vee, l)$ is sent to a lower semilattice by $\mathcal{S}$. Because $\mathcal{S}$ preserves finite products we have a semilattice $(\mathcal{S}X, \vee') \defeq (\mathcal{S}X, \synRef\vee)$. It remains to show that $S \pathle S \vee' T$ for all $S, T : \mathcal{S}X$. We have a subset $\Phi \subseteq (\synRef{} X)^2 \cong \synRef (X^2)$ defined by $\Phi(S, T)$ if and only if $S \pathle S \vee' T$. Thus we proceed by \synRef{}-induction. First, we observe that every $\eta(x)$ is in $\Phi$, as $\eta(x) \pathle \eta(x \vee y) = \eta(x) \vee' \eta(y)$ by the naturality of $\eta : X \to \synRef$. To see that $\Phi$ is a predomain we observe that it is the vertex of the following pullback diagram.

    \[
    \DiagramSquare{
        width = 3cm,
        nw = \Phi,
        sw = (\mathcal{S}X)^2,
        ne = {\mathcal{S}X^\I},
        se = {\mathcal{S}X^2}, 
        east = \partial,
        south = {\langle \pi_1, \vee' \rangle},
        nw/style = pullback,
    }
    \]
\end{proof}

\begin{cor}
    The composition $\kw{H} = \CCat \hookrightarrow \ECat \xrightarrow{H} \kw{PLSL}(\ECat) \xrightarrow{\synRef{}} \kw{PLSL}(\CCat) = \kw{LSL}$ is left adjoint to the forgetful functor $\kw{LSL} \hookrightarrow \CCat$
\end{cor}

\begin{restatable}{prop}{PropLowerPointed}
    If $X$ is pointed, then $H X$ is pointed. 
\end{restatable}

\begin{proof}
    Let $\bot : X$ be the distinguished point and $\bot' = \eta~\bot : H X$ be the candidate element. We show that there is a path $\eta\bot \leadsto S$ for all $S : H X$ by induction. 
    \begin{enumerate}
        \item \emph{Semilattice structure}. We may define a path $\bot' \leadsto \eta~x$ for all $x : X$ by postcomposing $\bot \leadsto x$ by $\eta : X \to H X$. For $S \vee T : H X$ we have by induction paths $\alpha : \bot' \leadsto S$ and $\beta:\bot' \leadsto T$, and we may define $\alpha \vee \beta : \bot' \leadsto S \vee T$ by sending $i : I$ to $\alpha~i \vee \beta~i$ and observing that $\alpha~0 \vee \beta~0 = \bot' \vee \bot' = \bot'$ by idempotence. The semilattice laws hold because they hold at each point in a path $\bot' \leadsto S$.
        \item \emph{Axiom $l : S \leadsto S \vee T$}. Fixing $S, T: H X$, we have by induction paths $\alpha : \bot' \leadsto S$ and $\beta:\bot' \leadsto T$. We need to construct a path $\gamma : \bot' \leadsto l(S, T, i)$ for every $i : \I{}$, \ie{} a two-dimensional path $\gamma : \I{}^2 \to H X$. We define $\gamma$ as follows. 

        \iblock{
            \mrow{\gamma(i, j) = l(\alpha~j, \beta~j, i)}
        }

        We verify the boundary conditions for all $i : \I$. First we have $\gamma(i, 0) = l(\alpha~0, \beta~0, i) = l(\bot', \bot', i) = \bot'$, where the last equality follows from the fact that $H X$ is boundary separated. Next we have $\gamma(i, 1) = l(\alpha~1, \beta~1, i) = l(S, T, i)$ as required. In addition we need to verify the two-dimensional boundary conditions of $\gamma(0, j) =_{\bot' \leadsto S} \alpha$ and $\gamma(1, j) =_{\bot' \leadsto S \vee T} \alpha \vee \beta$, which follow by similar direct computations. 
        \item \emph{Boundary separation}. This again follows since $H X$ is boundary separated. 
    \end{enumerate}
\end{proof}

\begin{cor}
    The composition $\LP = \CCat \xrightarrow{\kw{H}} \kw{LSL} \hookrightarrow \kw{SL}(\CCat)$ is a powerdomain construction. 
\end{cor}

An entirely analogous construction may be carried out to accommodate the axiom $S \vee T \le S$ by means of the quotient inductive type displayed in \cref{fig:lower-and-upper}, giving us the synthetic upper powerdomain construction. 

\begin{prop}
    Let \kw{USL} be the category of \emph{upper semilattices}, defined analogously to \cref{def:lowersemilattice}. There is a left adjoint functor $\CCat \to \kw{USL}$ and $\UP{} = \CCat \to \kw{USL} \hookrightarrow \kw{SL}(\CCat)$ is a powerdomain construction. 
\end{prop}

\section{Denotational semantics}

In this section we use the synthetic powerdomains to define models of a nondeterministic extension of PCF dubbed \pcfnd. For generality we consider a \emph{call-by-push-value}~\cite{levy:2001} refinement of \pcf{} inspired by the semantic free-forgetful adjunction between the category of predomains \CCat{} and the category of $\TT$-algebras for a suitable monad $\TT$ that is parameterized over the various powerdomain constructions of \cref{sec:powerdomains}. Recall that types in call-by-push-value languages bifurcate into \emph{value} types and \emph{computation} types. 

\iblock{
    \mrow{\textit{Value types}\; A  := \kw{nat} \mid \kw{U}(X)}
    \mrow{\textit{Computation types}\; X := \kw{F}(A) \mid A \to X}
}

In the above \kw{F} and \kw{U} are type operators corresponding to the semantic free and forgetful functors. Programs are similarly stratified into values or computations. We write $(\tpv{}, \tpc{})$ for the (ambient) type of (value, computation) types, and \tm{A} for the type of \pcfnd{} terms of value type $A : \tpv$; computations are similarly classified by the type \tm{U X}. 

\iblock{
  \mrow{\textit{Values}\; a  := \kw{zero} \mid \kw{suc}(a)}
  \mrow{\textit{Computations}\; M := \kw{ret}(v) \mid \kw{bind}(M, x.N) \mid \lambda x.M \mid M~a \mid \kw{ifz}(a, M, x.N) \mid \kw{fix}(x.M) \mid M \vee N}
}

The typing rules for the \pcf{}-fragment are completely standard; nondeterministic branching $\vee : \tm{\U{X}} \to \tm{\U{X}} \to \tm{\U{X}}$ is available for computational types $X$.

\subsection{Denotational semantics}\label{sec:den-sem}

We parameterize the model construction over a powerdomain construction $P : \CCat \to \kw{SL}(\CCat)$ and write $\PP = (P, \eta_\PP, \mu_\PP)$ for the associated monad on $\CCat$. By \cref{prop:monad} there is a distributive law $l : \LL\PP \to \PP\LL$ and we write $\TT : \CCat \to \CCat$ for the induced composite monad $\PP\circ\LL$. We define a model by specifying a predomain for each value type and a $\TT$-algebra for each computation type. Let us write \kw{Predom} for the subuniverse of predomains (of some unspecified universe $\mathcal{U}$) and $\kw{Alg}_\TT(\mathcal{U})$ for the type of $\TT$-algebras $\mathcal{X}$ whose carrier set $|\mathcal{X}|$ is valued in some universe $\mathcal{U}$; we also write $\alpha_\mathcal{X} : \TT|\mathcal{X}| \to |\mathcal{X}|$ for the associated algebra structure map. We aim to define the following four semantic maps $\sem{-}$ (by an abuse of notation). 

\begin{multicols}{2}
    \iblock{
  \mrow{\sem{-} : \tpv \to \kw{Predom}}
  \mrow{\sem{-} : (A : \tpv) \to \tm{A} \to \sem{A}}
}
\columnbreak
\iblock{
  \mrow{\sem{-} : \tpc \to \kw{Alg}_\TT(\kw{Predom})}
  \mrow{\sem{-} : (X : \tpc) \to \tm{\kw{U}X} \to |\sem{X}| }
}
\end{multicols}

To lighten the notation we will elide the inclusion of computations into values and just write \tm{X} for \tm{\kw{U}X}. 

\subsubsection{Values}

Although the ambient natural numbers object \Nat{} of the topos \ECat{} does not necessarily have to be a predomain, its synthetic reflection $\NatP \defeq \synRef\Nat$ is always the natural numbers object in the category of predomains. Setting \sem{\kw{nat}} = \NatP, the terms \kw{zero}, \kw{suc}, \kw{ifz} of \pcfnd{} may then be interpreted in the canonical way.  

\subsection{Computations}

The computation type $\kw{F}(A)$ is sent to the free $\TT$-algebra on $\sem{A}$; $\kw{ret}$ and \kw{bind} are interpreted in the usual way using the monadic structure of $\TT$. For function types, we use the general fact that in a Cartesian closed category with a monad $\MonIdent{M}$, the function space $D^C$ has a unique $\MonIdent{M}$-algebra structure whenever $D$ has one~\cite{levy:2001}, and send $A \to X$ to the $\TT$-algebra whose underlying set is $\sem{A} \to |\sem{X}|$. Then lambda abstraction and application are canonically interpreted using the universal property of exponential objects. For fixed-point computation $\kw{fix}_X : \tm{X} \to  \tm{X}$, we recall that $\TT$ is the composite monad $\PP\LL$ and so the $\TT$-algebra \sem{X} is a predomain carrying a canonical $\LL$-algebra structure, from which it follows we may interpret $\kw{fix}_X M$ as the unique fixed-point of $\sem{M} : |\sem{X}| \to |\sem{X}|$. Lastly, for nondeterministic branching, we need to interpret the term $\vee_X : \tm{X} \to \tm{X} \to \tm{X}$, which follows from the fact that every $\TT$-algebra is also a $\PP$-algebra and thus carries a semilattice structure.

\begin{prop}[{\cite[Theorem 3.6.3]{barr-wells:1985}}]\label[prop]{prop:semilattice-map}
  For every $\TT$-algebra $(X, \alpha)$, the algebra map $T X \to X$ is a $\PP$-algebra homomorphism and thus preserves the semilattice structure. 
\end{prop}

\section{Operational semantics}\label{sec:op-sem}

Whereas denotational semantics allows us to reason about programs in terms of domains it does not say how to execute programs; the purpose of operational semantics is to specify how to assign meaning to programs in a completely mechanical way. In structural operational semantics~\cite{plotkin:1981} this mechanical procedure is specified by a family of relations ${\mapsto_A} \mathrel{\subseteq} \tm{A} \times \tm{A}$ such that each of $M \mapsto M'$ is a decidable proposition. The fundamental coherence property of denotational and operational semantics for ordinary \pcf{} is called \emph{computational adequacy}, and it states that at \emph{base type}, say at natural numbers, we have that $\sem{M} = \eta_\LL(n)$ implies $M \mapsto^* \ret{\overline{n}}$; the converse is known as \emph{soundness}. While the proof of soundness for both \pcf{} and \pcfnd{} are analogous to proofs in the analytic setting the proof of computational adequacy demands the evaluation relation to be \emph{observable}, \ie{} an \I{}-proposition. Following prior work~\cite{niu-sterling-harper:2024} we will employ an intermediate semantics dubbed the \emph{observational semantics}, which in our setting provides the additional benefit of factoring out structures that are common to all three powerdomains.   

\subsection{May and must evaluation}\label{sec:evaluation}

First note that $M \Downarrow v$ is observable if \I{} is \emph{closed} under countable joins of decidable propositions. Indeed since each $M \mapsto M'$ is decidable, $M \Downarrow v$ is equivalent to the join of the sequence of propositions $(M \mapsto^{(n)} M') \land (M' = \ret{v})$; that \I{} is closed under such joins means that whenever $\sup \phi_i = \top$ there exists $i : \Nat$ such that $\phi_i = \top$. Under the axioms we set out \I{} will always have countable joins, but using them to \emph{define} $M \Downarrow v$ would lead to a strange evaluation relation: $M \Downarrow v$ need not imply there exists some finite sequence of transitions $M \mapsto^{*} \kw{ret}(v)$. We may say that such joins lack the \emph{disjunction} or \emph{existence} property.

\begin{defn}[Axiom CDJ]
    We say \I{} is closed under countable, decidable propositions when it has such joins and they are preserved by $(- = \top) : \I{} \hookrightarrow \Omega$. In other words we have that \I{} has countable, decidable joins $\vee_{i : \Nat} \phi_i$ that are computed as $\exists {i : \Nat}.~\phi_i$.
\end{defn}

\begin{ax}\label[ax]{ax:cdj}
  Axiom CDJ holds.
\end{ax}

We may now proceed to define the evaluation relation for \pcfnd{}. The small-step semantics of \pcfnd{} just extends the standard semantics of \pcf{} with $\choice{M}{N} \mapsto M$ and $\choice{M}{N} \mapsto N$, consequently rendering evaluation nondeterministic and resulting in \emph{may} and \emph{must} evaluation. Because only base types are observable we restrict attention to the type of natural numbers. 

\begin{defn}
    Given $M : \tm{\F{\kw{nat}}}$ and $v : \kw{nat}$, \emph{may evaluation} $M \Downarrow v$ holds when there exists a transition sequence $M \mapsto^* \kw{ret}(v)$.
\end{defn}

To define must evaluation we need to keep track of all possible transitions from a given program; because there is a finite number of such candidates we define a family of relations ${\nextS{A}} \mathrel{\subseteq} \tm{A} \times (1 + K(\tm{A}))$, writing ${\nextS{}}$ when the type is clear.

\begin{defn}
    The relation ${\nextS{}}$ is defined by the following rules.
    \begin{mathpar}
        \inferrule*{
            ~
        }{
            \ret{v} \nextS{} \varnothing
        }

        \inferrule*{
            ~
        }{
            \kw{bind}(\ret{v}, x.N) \nextS{} \MkSet{N[v/x]}
        }

        \inferrule*{
            M \ne \ret{v}\\
            M \nextS{} S
        }{
            \kw{bind}(M, x.N) \nextS{} \MkSet{\kw{bind}(M', x.N) \mid M' \in S}
        }

        \inferrule*{
            ~
        }{
            \lambda x.M \nextS{} \varnothing
        }

        \inferrule*{
            ~
        }{
            (\lambda x.M)~a \nextS{} \MkSet{M[a/x]}
        }

        \inferrule*{
            M \nextS{} S \\
            M \ne \lambda x.N
        }{
            M~a \nextS{} \MkSet{M'~a \mid M' \in S}
        }

        \inferrule*{
            ~
        }{
            \kw{ifz}(\kw{zero}, M, x.N) \nextS{} \MkSet{M}
        }

        \inferrule*{
            ~
        }{
            \kw{ifz}(\kw{suc}(a), M, x.N) \nextS{} \MkSet{N[a/x]}
        }

        \inferrule*{
            ~
        }{
            \kw{fix}(x.M) \nextS{} \MkSet{M[\kw{fix}(x.M)/x]}
        }

        \inferrule*{
            ~
        }{
            M \vee N \nextS{} \MkSet{M, N}
        }
    \end{mathpar}
\end{defn}

The relation ${\nextS{A}}$ tracks precisely the possible one-step transitions.

\begin{prop}\label[prop]{prop:nextS-mapsto} 
    We have that $M \nextS{} S$ and $M \nextS{} S'$ implies that $S = S'$ and that $M \nextS{} \MkSet{M' : \tm{A} \mid M \mapsto M'}$.
\end{prop}

\begin{defn}
    \emph{Must evaluation} ${\must}$ is the least relation on $\tm{\F{\kw{nat}}} \times K(\tm{\kw{nat}})$ closed under the following.

    \begin{enumerate}
        \item We have that $\ret{n} \must \MkSet{n}$. 
        \item If $M \nextS{} \MkSet{M_1, \dots, M_k}$ and $M_i \must{} S_i$ for all $i \in \MkSet{1, \dots, k}$, then $M \must{} \bigcup_i S_i$. 
    \end{enumerate}
\end{defn}
\begin{prop}\label[prop]{prop:must-prop}
    If $M \must{} S$ and $M \must{} S'$ then $S = S'$. 
\end{prop}

\begin{prop}
    Assuming \cref{ax:cdj}, we have that $M \Downarrow v$ and $M \must{} S$ are \I{}-propositions. 
\end{prop}

\section{Observational semantics}\label{sec:ob-sem}

While defining the operational semantics seems to necessitate certain properties about the dominance \I{} that need not hold in some models of synthetic domain theory~\cite{van-oosten:2000,sterling-harper:2022}, one may always define an \emph{observational semantics} that is halfway between the denotational and operational semantics as shown in Niu, Sterling and Harper~\cite{niu-sterling-harper:2024}. The idea is that one may define a function $\kw{run} : \tm{\F{\kw{nat}}} \to T\NatP$ that directly ``executes'' programs of base type by recursively unfolding the small-step semantics. In \opcit{} it was argued that the observational semantics (called the \emph{computational semantics} in that work) is more natural in synthetic domain theory and speculated how one may relate it to traditional operational semantics from an \emph{external} point of view without assuming \cref{ax:cdj}. In our setting the observational semantics will also be parameterized by a given powerdomain construction and thus has the additional benefit of factoring out commonalities of the soundness and adequacy proofs for the lower, upper, and convex semantics of \pcfnd{}. As in \cref{sec:den-sem} we assume there is a powerdomain construction $\mathcal{P} : \CCat \to \kw{SL}(\CCat)$ and write $\TT = (T, \eta_\TT, \mu_\TT)$ for the composite monad $\PP\circ\LL$.

\begin{prop}
    There is a function $\kw{next}_A : \tm{A} \to 1 + K(\tm{A})$ such that $M \nextS{A} (\kw{next}_A~M)$ holds.
\end{prop}

\begin{defn}
    The \emph{observational semantics} $\kw{run}$ is defined to be the fixed point of the following functional. 

    \iblock{
        \mrow{F_\kw{run} : (\tm{\F{\kw{nat}}} \to T\NatP) \to (\tm{\F{\kw{nat}}} \to T\NatP)}
        \mrow{F_\kw{run}~f~M = \kw{case}~\kw{next}_{\F{\kw{nat}}}(M) \begin{cases}
            \varnothing &\hookrightarrow \sem{M}\\
            \MkSet{M_1, \dots, M_k} &\hookrightarrow f~M_1 \vee \dots \vee f~M_k
        \end{cases}}
    }
    
\end{defn}

\subsection{Observational soundness}

We have that the denotational semantics is sound with respect to the observational semantics. First observe that the denotational semantics is invariant under set of possible one-step transitions of a term.

\begin{prop}\label[prop]{prop:sound-one-step}
    Given that $M \nextS{} S$ holds for some nonempty set of terms $S = \MkSet{M_1, \dots, M_k}$, we have that $\sem{M} = \sem{M_1} \vee \dots \vee \sem{M_k}$. 
\end{prop}

\begin{proof}
    The proof proceeds by induction on the derivation of $M \nextS{} S$. For the cases in which $S$ is a singleton the result follows by directly computing the denotational semantics. Otherwise we have two nonvacuous cases. 

    \begin{enumerate}
    \item Suppose that $\kw{bind}(M, x.N) \nextS{} \MkSet{\kw{bind}(M', x.N) \mid M' \in S}$ because $M \nextS{} S$. By inspection we must have that $S = \MkSet{M_1, \dots, M_k}$ is a nonempty set. We compute. 

    \iblock{
        \mhang{\sem{\kw{bind}(M, x.N)} = \sem{x.N}^\dagger\sem{M}}{
            \mrow{= \sem{x.N}^\dagger(\sem{M_1} \vee \dots \vee \sem{M_k})}
            \mrow{= \sem{x.N}^\dagger\sem{M_1} \vee \dots \vee \sem{x.N}^\dagger\sem{M_k}}
            \mrow{= \sem{\kw{bind}(M_1, x.N)} \vee \dots \vee \sem{\kw{bind}(M_k, x.N)}}
        }
    }

    In the above for a map $f : A \to X$ where $X$ has the structure of a $\TT$-algebra $(X, \alpha)$ we write $f^\dagger : TA \to X$ for the map $TA \xrightarrow{T f} TX \xrightarrow{\alpha} X$, which preserves semilattice structure by \cref{prop:semilattice-map}. 
    \item Suppose that $M~a \nextS{} \MkSet{M'~a \mid M' \in S}$ because $M \nextS{} S$ for some nonempty $S = \MkSet{M_1, \dots, M_k}$. We compute.

    \iblock{
        \mhang{\sem{M~a} = \sem{M}~\sem{a}}{
            \mrow{= (\sem{M_1} \vee \dots \vee \sem{M_k})~\sem{a}}
            \mrow{= \sem{M_1}~\sem{a} \vee \dots \vee \sem{M_k}~\sem{a}}
            \mrow{= \sem{M_1~a} \vee \dots \vee \sem{M_k~a}}
        }
    }

    In the above we use the fact that the semilattice structure on function types is defined pointwise.
    \end{enumerate}
\end{proof}

\begin{restatable}[Observational soundness]{thm}{ThmObservationalSoundness}\label[thm]{thm:observational-soundness}
    We have that $(\kw{run}~M) \pathle \sem{M}$ for all $M : \tm{\F{\kw{nat}}}$.   
\end{restatable}

\begin{proof}
    We proceed by fixed-point induction. Consider the following predicate $\Phi \subseteq (\tm{\F{\kw{nat}}} \to T\NatP)$. 

    \iblock{
        \mrow{\Phi = \MkSet{f \mid \forall M.~f~M \pathle \sem{M}}}
    }
    
    By \cref{prop:complete,cor:intersection} we have that $\Phi$ is admissible. It remains to show that $F_\kw{run} \Phi \subseteq \Phi$. We fix $M : \tm{\F{\kw{nat}}}$ and case on $F_\kw{run}~f~M$. 
    \begin{enumerate}
        \item We have that $\kw{next}(M) = \varnothing$. Then by definition we need to show that $F_\kw{run}~f~M = \sem{M} \pathle \sem{M}$ which holds. 
        \item We have that $\kw{next}(M) = \MkSet{M_1, \dots, M_k}$. Then we need to show that $f~M_1 \vee \dots \vee f~M_k \pathle \sem{M}$. By \cref{prop:sound-one-step} it suffices to show $f~M_i \pathle \sem{M_i}$ for all $M_i$, but this is the assumption $f \in \Phi$. 
    \end{enumerate}

\end{proof}

\subsection{Observational adequacy}

The proof of the converse property of \emph{observational adequacy}, that $\sem{M} \pathle (\kw{run}~M)$ for all $M : \tm{\F{\kw{nat}}}$, follows from a logical relations argument. We follow the approach of Niu, Sterling, and Harper~\cite{niu-sterling-harper:2024} to lift a binary semantic-syntactic relation on $\sem{A} \times \tm{A}$ to a relation on $\sem{T A} \times \tm{\F{A}}$, which we show to be an instance of the semantic $\top\top$-lifting construction of Katsumata~\cite{katsumata:2005}, where the construction is parameterized in a \emph{result type} $G$ and \emph{result predicate}\footnote{In Katsumata~\cite{katsumata:2005} the result type may be any type, but restricting to syntactic types is sufficient for our purposes. } $\Phi_G \subseteq T\sem{G} \times \tm{\F{G}}$. 

\begin{defn}
    Given a result type $G$ and result predicate $\Phi_G$, we define the following liftings of a relation $R \subseteq \sem{A} \times \tm{A}$ for any \pcfnd{} type $A$. 

    \iblock{
        \mrow{R^\top \subseteq (\sem{A} \to T\sem{G}) \times \tm{A \to \F{G}}; R^{\top\top} \subseteq T\sem{A} \times \tm{\F{A}}}
        \mrow{R^\top = \MkSet{(f, f') \mid \forall a, a'.~ R(a,a') \to \Phi_G(f~a, f'~a')}}
        \mrow{R^{\top\top} = \MkSet{(M, M') \mid \forall f, f'. R^\top(f, f') \to \Phi_G(f^\dagger M, \bind{M'}{f'})}}
    }

    Recall that $f^\dagger$ is the monadic lift of a function $f$. The relation $R^{\top\top}$ is called the \emph{$\top\top$-lifting} of $R$.
\end{defn}

Using $\top\top$-lifting we may construct a logical relation that is parameterized over a result type and predicate, the choice of which is motivated by the particular metatheorem to be proved. 

\begin{defn}
    The \emph{formal approximation relations} are the following family of relations ${\lhd_A} \subseteq \sem{A} \times \tm{A}$ defined by recursion on the structure of syntactic types $A$.
    
    \iblock{
        \mrow{n \lhd_\kw{nat} n' \defeq \sem{n'} = n} 
        \mrow{M \lhd_{\U{\F{A}}} M' \defeq M \lhd_A^{\top\top} M'}
        \mrow{f \lhd_{\U{(A \to X)}} f' \defeq \forall a, a'.~ (a \lhd_A a') \to (f~a) \mathrel{\lhd_X} (f'~a')}
    }
    
\end{defn}

The proof of observational adequacy proceeds in two steps: 1) we prove the \emph{fundamental lemma}, which states that $\sem{M} \lhd_A M$ for all $M : \tm{A}$ and 2) pick a result type and predicate such that $\sem{M} \lhd_{\U{\F{\kw{nat}}}} M$ that implies $\sem{M} \pathle (\kw{run}~M)$. The proof of the fundamental lemma proceeds by induction on the structure of \pcfnd{} terms; to prove the case for fixed-points we require that each formal approximation relation restricts to an admissible subset $- \lhd_X M$ for all computation types $X$ and $M : \tm{\U{X}}$.  

\begin{prop}
    If $- \lhd_{\U{X}} M$ is admissible for all $M : \tm{\U{X}}$, then so is $-\lhd_{\U{(A \to X)}} f$ for all $f : \tm{A \to X}$.  
\end{prop}

\begin{prop}
    If the result predicate $\Phi_G(-, M) \subseteq T\sem{G}$ is admissible for every $M : \tm{\F{G}}$, then so is $R^{\top\top}$ for every $R \subseteq \sem{A} \times \tm{A}$. 
\end{prop}

\begin{cor}
    We have that $- \lhd_{\U{\F{A}}} M$ is admissible for all $M : \tm{\F{A}}$. 
\end{cor}

\begin{restatable}{thm}{FLLR}
    For all $M : \tm{A}$, we have $\sem{M} \lhd_{A} M$.  
\end{restatable}

The proof of the fundamental lemma is completely standard. We are now in a position to pick a suitable admissible result predicate. 

\begin{prop}
    The result predicate $\Phi_\kw{nat} = \MkSet{(M, M') \mid M \pathle (\kw{run}~M')}$ restricts to an admissible subset $\Phi_\kw{nat}(-, M') \subseteq T\NatP$. 
\end{prop}

\begin{cor}[Observational adequacy]
    We have that $\sem{M} \pathle (\kw{run}~M)$ for all $M : \tm{\F{\kw{nat}}}$. 
\end{cor}

Together with \cref{thm:observational-soundness} we see that the observational and denotational semantics coincide. 

\begin{thm}\label{thm:coherence}
    We have that $(\kw{run}~M) = \sem{M}$ for all $M : \tm{\F{\kw{nat}}}$. 
\end{thm}

\section{Computational adequacy}\label{sec:adequacy}

In this section we relate the may and must evaluation relations defined in \cref{sec:evaluation} constituting the operational semantics of \pcfnd{} to the three denotational semantics obtained by instantiating the model construction in \cref{sec:den-sem} at the lower, upper, and convex powerdomains. By virtue of \cref{thm:coherence} we may instead work in terms of the observational semantics, which enables us to obtain the desired results by means of fixed-point induction. In the following we will also make use of constructions that make sense for any predomain $X$ with \I{}-equality, which we may instantiate to \NatP{} by virtue of \cref{prop:nat}. 

\subsection{Adequacy for the lower powerdomain construction}

For the lower powerdomain, we wish to show correspondence with respect to may evaluation, in the sense that every possible denotational result is possible operationally. In general, the membership relation in the various powerdomains $\mathcal{P}X$ can only be valued in \I{}-propositions when $X$ has \I{}-equality, which obstructs the use of fixed-point induction as a proof principle. However in the context of operational semantics we are only interested in querying the membership of \emph{defined} numbers, and the characteristic map of this restricted membership relation does factor through \I{}. 

\begin{defn}
    Let $?{=}  : L X \to \I^X$ be the \emph{is-defined-and-equal} operation defined by sending $(u, x)$ to $u{\downarrow} \mathbin{\angle} (u = x)$. 
\end{defn}

Because \I{} is closed under finite joins (by \cref{ax:finite-join}) we may view $\I{}$ as a lower semilattice as $(i \pathle i \vee j) \iff ((i = \top) \to (i \vee j = \top))$ always holds. This allows us to extend the operation $?{=} : L X \to \I^X$ to the lower powerdomain. 

\begin{defn}\label[defn]{def:def-member}
    The \emph{defined-membership relation} $\in? : X \to \LP{}(L X) \to \I{}$ is given by the unique lower semilattice homomorphism $\LP{}(L X) \to \I^X$ extending $?{=}$, viewing $\I^X$ as a product lower semilattice. 
\end{defn}

\begin{restatable}[Soundness]{thm}{ThmLowerSoundness}\label[thm]{thm:soundness-lower}
    For all $M : \tm{\F{\kw{nat}}}$ and $n : \tm{\kw{nat}}$, we have that $M \Downarrow n$ implies $\sem{n} \mathrel{\in?} \sem{M}$. 
\end{restatable}

\begin{proof}
    Assume $M \Downarrow n$ for some $n : \tm{\kw{nat}}$. We proceed by induction on the derivation of $M \Downarrow n$. 
    \begin{enumerate}
        \item Suppose that $M = \ret{n}$ for some $n : \tm{\kw{nat}}$. We need to show $\sem{n} \mathrel{\in?} \sem{\ret{n}} = \singleton{\eta_\LL(\sem{n})}$, which follows by the definition of defined membership. 
        \item Otherwise we have $M \mapsto M'$ and $M' \Downarrow n$. Note this implies that $\kw{next}(M) = \MkSet{M_1, \dots, M_k}$ with $M' = M_i$ for some $i \in \MkSet{1, \dots, k}$. By \cref{thm:coherence} it suffices to show $\sem{n} \mathrel{\in?} \kw{run}(M)$. Unfolding the observational semantics, we have $\kw{run}(M) = \kw{run}(M_1) \vee \dots \vee \kw{run}(M_k)$. By the universal property of defined membership we have $\sem{n} \mathrel{\in?} (\kw{run}(M_1) \vee \dots \vee \kw{run}(M_k))$ holds if and only if there exists $M_j$ such that $\sem{n} \mathrel{\in?} \kw{run}(M_j)$, and so the result holds by taking $M_j = M'$ and the inductive hypothesis. 
    \end{enumerate}
\end{proof}
    
We may also prove the expected converse statement. 

\begin{restatable}[Adequacy]{thm}{ThmLowerAdequacy}\label[thm]{thm:adequacy-lower}
    For all $M : \tm{\F{\kw{nat}}}$ and $u : \NatP$, $u \mathrel{\in?} \sem{M}$ implies there exists some $n : \tm{\kw{nat}}$ such that $M \Downarrow n$ and $\sem{n} = u$. 
\end{restatable}

\begin{proof}
    By \cref{thm:coherence} it suffices to prove the statement by replacing $\sem{M}$ with $\kw{run}(M)$. We proceed by fixed-point induction. By \cref{prop:nat} we have that $\Sigma_{n:\tm{\kw{nat}}}.~M \Downarrow n \land \sem{n} = u$ is equivalent to the proposition $\phi_{M,u} \defeq \exists n', k : \Nat. (M \mapsto^{(k)} \kw{ret}(\overline{n'})) \land (\eta_\synRef(n') = u)$, which is an \I{}-proposition by \cref{ax:cdj}. Consider the following subset $\Phi \subseteq (\tm{\F{\kw{nat}}} \to \LP(L \NatP))$. 

    \iblock{
        \mrow{\Phi = \MkSet{f \mid \forall M:\tm{\F{\kw{nat}}}, u : \NatP.~ u \mathrel{\in?} f(M) \to \phi_{M,u}}}
    }
    
    We observe that $\Phi = \bigcap_{M, u} \MkSet{f \mid u \mathrel{\in?} f(M) \to \phi_{M,u}}$ and that the \I{}-subset $\MkSet{f \mid u \mathrel{\in?} f(M)}$ \emph{does not} contain the least element $\lambda -.~\singleton{\bot}$. Thus we have that $\Phi$ is admissible by \cref{cor:intersection,prop:admissible}. It remains to show that $F_\kw{run} \Phi \subseteq \Phi$. We case on $\kw{next}(M)$. 
    \begin{enumerate}
        \item We have that $\kw{next}(M) = \varnothing$ and so $M = \kw{ret}(n)$ for some $n : \tm{\kw{nat}}$. Then $F_\kw{run}~f~M = \sem{M} = \sem{\ret{n}} = \singleton{\eta_\LL(\sem{n})}$. Suppose $u \mathrel{\in?} \singleton{\eta_\LL(\sem{n})}$. Then by definition of defined membership we have that $u = \sem{n}$, and so the result holds since $\ret{n} \Downarrow n$. 
        \item We have that $\kw{next}(M) = \MkSet{M_1, \dots, M_k}$. Then $F_\kw{run}~f~M = (f~M_1) \vee \dots \vee (f~M_k)$. Suppose $u \mathrel{\in?}  (f~M_1) \vee \dots \vee (f~M_k)$. By definition this implies that $u \mathrel{\in?}  (f~M_i)$ for some $M_i$. By the induction hypothesis we have that $M_i \Downarrow n$ for some $n : \tm{\kw{nat}}$ such that $\sem{n} = u$. Thus the result holds since $M \mapsto M_i$ by definition of $\kw{next}(M)$. 
    \end{enumerate}
\end{proof}

\subsection{Adequacy for the upper powerdomain construction}

In contrast to the lower powerdomain, the upper powerdomain does not admit a membership relation factoring through \I{} even if the base domain has \I{}-equality, essentially because $(\I{}, \vee)$ is not an upper semilattice. Instead we aim to show that a computation at base type $M : \UP(L \NatP)$ is either divergent or terminates to an element of the upper semilattice $\sem{M} : \UP(\NatP)$. We aim to show that if $\sem{M}$ is defined then $M$ must-evaluates to the set of numbers denoted by \sem{M}. First we observe that the dominance $\I{}$ has the structure of an upper semilattice with the semilattice operation given by logical conjunction $i \land j \defeq i \mathbin{\angle} (\lambda -.~ j)$. 

\begin{defn}
    The \emph{total termination support} of an element $S : \UP(L X)$ is given by the unique upper semilattice homomorphism $\Supp{} : \UP(L X) \to \I$ extending the ordinary termination support $\supp{} : L X \to \I$, viewing $(\I{}, \land)$ as an upper semilattice. 
\end{defn}

Next we show that the upper powerdomain on $X$ embeds into the upper powerdomain on the lift of $X$ whose image contains only totally defined elements. Given a predicate $\phi : A \to \Omega$, we write $\MkSet{a : A \mid \phi(a)} \defeq \Sigma_{a : A}. \phi(a)$ for the type of elements of $A$ satisfying $\phi$. 

\begin{restatable}{prop}{PropTotalToPartial}
    There is a map $\iota : \UP(X) \to \MkSet{S : \UP(L X) \mid S\Supp{}}$ such that $\pi_1\iota = \UP\eta_\LL$. 
\end{restatable}

\begin{proof}
    Consider the image of the total termination support map on $\UP(X)$. 

    \[
    \begin{tikzpicture}[diagram]
        \SpliceDiagramSquare<sq/>{
            width = 3cm, 
            nw = X, 
            sw = \UP(X), 
            ne = L X,
            se = \UP(L X),
            north = \eta_\LL,
            west = {\eta_\UP{}},
            south = \UP\eta_\LL,
            east = \eta_\UP,
        }
        \node[right = of sq/se, xshift = 0.5cm](S) {$\I$};
        \draw[->] (sq/ne) to node[above] {$\supp$} (S);
        \draw[->] (sq/se) to node[below] {$\Supp$} (S);
    \end{tikzpicture}
    \]

    Because the composition $\Supp \circ \UP\eta_\LL$ is an upper semilattice homomorphism extending $\supp \eta_\LL = \top$, we have that it must be the constant homomorphism determined by $\top$. Thus we may define the desired map $\UP(X) \to \MkSet{S : \UP(L X) \mid S\Supp{}}$ by sending $S$ to $((\UP\eta_\LL)S, \triv)$. 
\end{proof}

\begin{restatable}{prop}{PropUpperSemilattice}\label[prop]{prop:upper-semilattice}
Given an upper semilattice $Y$, we have that $\Sigma_{S : \UP(L X)}.~Y^{S\Supp{}}$ is an upper semilattice.
\end{restatable}

\begin{proof}
    First we observe that $\Sigma_{S : \UP(L X)}.~Y^{S\Supp{}}$ is a predomain as it is the vertex of the following pullback diagram of predomains. 
    
    \[
    \DiagramSquare{
        width = 3cm,
        nw = \Sigma_{S : \UP(L X)}.~Y^{S\Supp{}},
        sw = \UP(L X),
        ne = L Y,
        se = \I,
        east = \supp,
        south = \Supp,
        nw/style = pullback,
    }
    \]
    
    For the semilattice operation we set $(S, p) \vee (T, q)$ to be $(S \vee T, p \vee q)$ where we define $p \vee q : (S \vee T) \Supp \to Y$ by restricting $p : S\Supp \to Y$ and $q : T \Supp \to Y$ as follows. 
    
        \iblock{
            \mrow{p' = p \circ (S\Supp \land T\Supp \to S\Supp)}
            \mrow{q' = q \circ (S\Supp \land T\Supp \to T\Supp)}
        }

    Here $e \circ (\psi \to \phi): \psi \to Y$ is the composition of a partial element $e : \phi \to Y$ along an implication $\psi \to \phi$. Since $\Supp$ is an upper semilattice homomorphism we have that $(S \vee T)\Supp = S\Supp \land T\Supp$, we may define $p \vee q$ as $\lambda u:(S \vee T)\Supp.~p'(u) \vee q'(u)$. This operation clearly satisfies the semilattice algebra laws. It remains to show that there is a path 

    \iblock{
        \mrow{(S, p) \vee (T, q) = (S \vee T, p \vee q) \leadsto (S, p)}
    }

    for all $(S, p), (T, q) : \Sigma_{S : \UP(L X)}.~Y^{S\Supp{}}$. We may construct this path by constructing the following paths. 
    
    \iblock{
        \mrow{\alpha : (S \vee T, p \vee q) \leadsto (S \vee T, p \circ (S\Supp \land T\Supp \to S\Supp))}
        \mrow{\beta : (S \vee T, p \circ (S\Supp \land T\Supp \to S\Supp)) \leadsto (S, p)}
    }

    \begin{enumerate}
        \item To define $\alpha : \I \to \Sigma_{S : \UP(L X)}.~Y^{S\Supp{}}$ we send $i : \I$ to $(S \vee T, \lambda - : (S \vee T)\Supp.~?)$, where we are required to construct an element $? : Y$ under the assumption that $(S \vee T)\Supp = S\Supp \land T\Supp$ holds. Thus supposing $u : S\Supp \land T\Supp$, we may define the requisite element to be $l_Y(p'(u), q'(u), i)$, where $l_Y(y, y', -) : y \vee y' \pathle y$ is the (necessarily unique) path asserted by the assumption that $Y$ is an upper semilattice. We first check that the lower boundary is $(S \vee T, p \vee q)$.
        
        \iblock{
            \mrow{(\pi_1\alpha)(0) = (S \vee T) }
            \mrow{}
            \mhang{(\pi_2\alpha)(0) = \lambda u:(S \vee T)\Supp.~l_Y(p'(u), q'(u), 0)}{
                \mrow{= \lambda u:(S \vee T)\Supp.~p'(u) \vee q'(u)}
                \mrow{= p \vee q \qquad \text{(by definition)}}
            }
        }

        Next we check that the upper boundary is $(S \vee T, p')$.

        \iblock{
            \mrow{(\pi_1\alpha)(1) = (S \vee T)}
            \mrow{}
            \mhang{(\pi_2\alpha)(1) = \lambda u:(S \vee T)\Supp.~l_Y(p'(u), q'(u), 1)}{
                \mrow{= \lambda u:(S \vee T)\Supp.~p'(u)}
                \mrow{= p'}
            }
        }

        \item Note that by construction of the free upper semilattice we have a unique path $\gamma : S \vee T \leadsto S$. Because the path order and implication order coincide on $\I$, we have that $(\gamma~i)\Supp \to (\gamma~1)\Supp$ and thus $(\gamma~i)\Supp \to (\gamma~1)\Supp = S\Supp$ for all $i : \I$. We define $\beta$ by sending $i : \I$ to the following element of $\Sigma_{S : \UP(L X)}.~Y^{S\Supp{}}$. 
        
        \iblock{
            \mrow{(\gamma~i, \lambda u: (\gamma~i)\Supp.~(p \circ ((\gamma~i)\Supp \to S\Supp))(u))}
        }

        We compute the lower boundary to be $(S \vee T, p')$. 

        \iblock{
            \mrow{(\pi_1\beta)(0) = \gamma~0 = (S\vee T)}
            \mrow{}
            \mhang{(\pi_2\beta)(0) = \lambda u: (\gamma~0)\Supp.~(p \circ ((\gamma~0)\Supp \to S\Supp))(u)}{
                \mrow{= \lambda u: (S\vee T)\Supp.~(p \circ ((S\vee T)\Supp \to S\Supp))(u)}
                \mrow{= p'}
            }
        }

        Lastly, the upper boundary is $(S, p)$ as required.

        \iblock{
            \mrow{(\pi_1\beta)(1) = \gamma~1 = S}
            \mrow{}
            \mhang{(\pi_2\beta)(1) = \lambda u: (\gamma~1)\Supp.~(p \circ ((\gamma~1)\Supp \to S\Supp))(u)}{
                \mrow{= \lambda u: S\Supp.~(p \circ (S\Supp \to S\Supp))(u)}
                \mrow{= p}
            }
        }
    \end{enumerate}
\end{proof}

\begin{restatable}{prop}{PropHomomorphism}\label[prop]{prop:upper-homomorphism}
The embedding $\iota : \UP(X) \to \MkSet{S : \UP(L X) \mid S\Supp{}}$ is a semilattice homomorphism. 
\end{restatable}

\begin{proof}
    By \cref{prop:upper-semilattice} we have that \MkSet{S : \UP(L X) \mid S\Supp{}} is an upper semilattice with semilattice operation defined as $(S, \triv : S\Supp) \vee (T, \triv : T\Supp) = (S \vee T, \triv : (S\vee T)\Supp)$. Thus the result follows as $\iota(S \vee T) = ((\UP\eta_\LL)(S \vee T), \triv) = ((\UP\eta_\LL)S \vee (\UP\eta_\LL)T, \triv) = \iota(S) \vee \iota(T)$. 
\end{proof}

\begin{restatable}{thm}{ThmUpperCharacterization}
We have an isomorphism $\iota : \UP(X) \to \MkSet{S : \UP(L X) \mid S\Supp{}}$. 
\end{restatable}

\begin{proof}
    By \cref{prop:upper-homomorphism} it suffices to show that $\MkSet{S : \UP(L X) \mid S\Supp{}}$ is the free upper semilattice on $X$, which is to show that it has the unique extension property along $X \xhookrightarrow{\iota \eta_\UP} \MkSet{S : \UP(L X) \mid S\Supp{}}$ for any upper semilattice $Y$ as follows. 

    \[\Local[f][->]{X \cong \MkSet{u : L X \mid u\supp}}{Y}{\MkSet{S : \UP(L X) \mid S\Supp{}}}[][\overline{f}][3cm]\]

    By exponential transpose it suffices to show that every $f : (u : L X) \to Y^{(\eta_\UP(u))\Supp}$ extends uniquely along $\eta_{\UP} : L X \to \UP(L X)$ to a homomorphism $\overline{f} : (S : \UP(L X)) \to Y^{S\Supp}$. By \cref{prop:upper-semilattice} we have that every $g : L X \to \Sigma_{S : \UP(L X)}.~Y^{S\Supp}$ extends uniquely to a homomorphism $\overline{g} : \UP(L X) \to \Sigma_{S : \UP(L X)}.~Y^{S\Supp}$ as follows. 

    \[
    \Local[g][->]{L X}{\Sigma_{S : \UP(L X)}.~Y^{S\Supp}}{\UP(L X)}[\eta_{\UP}][\overline{g}][3cm]
    \]

    Note that $\overline{g}$ is a section of $\pi_1 : \Sigma_{S : \UP(L X)}.~Y^{S\Supp} \to \UP(L X)$ when $g$ is a local section, and so every $f : (u : L X) \to Y^{(\eta_\UP(u))\Supp}$ has a unique homomorphism extension $\overline{f} : (S : \UP(L X)) \to Y^{S\Supp}$ along $\eta_{\UP}$. 
\end{proof}

Whenever the total termination support holds for some $S : \UP(L X)$ we will write $S : \UP(X)$ for the element determined by the isomorphism $\iota$. 

\begin{restatable}[Soundness]{thm}{ThmSoundnessUpper}\label[thm]{thm:soundness-upper}
    If $M \must{} \MkSet{n_1, \dots, n_k}$ then $\sem{M}\Supp{}$ holds and $\sem{M} = \singleton{\sem{n_1}} \vee \dots \vee \singleton{\sem{n_k}}$.
\end{restatable}

\begin{proof}
    By induction on the derivation of must evaluation. 
    \begin{enumerate}
        \item Suppose that $\ret{n} \must{} \MkSet{n}$. Then we have that $\sem{\ret{n}} = \singleton{\eta_\LL(\sem{n})} = \iota(\singleton{\sem{n}})$, as required.
        \item Otherwise we have $M \nextS{} \MkSet{M_1, \dots, M_k}$ and $M_i \must{} \MkSet{n^i_1, \dots, n^i_{j_i}}$ for all $i \in \MkSet{1, \dots, k}$. By induction we have that $\sem{M_i}\Supp$ holds and $\sem{M_i} = \iota(\singleton{\sem{n^i_1}} \vee \dots \vee \singleton{\sem{n^i_{j_i}}})$ for all $i$. We want to show that $\sem{M} = \iota(\bigvee_{i} \MkSet{\singleton{\sem{n^i_1}}, \dots, \singleton{\sem{n^i_{j_i}}}})$. By \cref{thm:coherence} it suffices to show that $\kw{run}(M) = \iota(\bigvee_{i} \MkSet{\singleton{\sem{n^i_1}}, \dots, \singleton{\sem{n^i_{j_i}}}})$. But this follows since 
        
        \iblock{
            \mhang{\kw{run}(M) = \kw{run}(M_1) \vee \dots \vee \kw{run}(M_k)}{
                \mrow{= \iota(\singleton{\sem{n^1_1}} \vee \dots \vee \singleton{\sem{n^1_{j_1}}}) \vee \dots \vee \iota(\singleton{\sem{n^k_1}} \vee \dots \vee \singleton{\sem{n^k_{j_k}}})}
                \mrow{= \iota(\singleton{\sem{n^1_1}} \vee \dots \vee \singleton{\sem{n^1_{j_1}}} \vee \dots \vee \singleton{\sem{n^k_1}} \vee \dots \vee \singleton{\sem{n^k_{j_k}}}) \; \text{(By \cref{prop:upper-homomorphism})}}
                \mrow{= \iota(\bigvee_{i} \MkSet{\singleton{\sem{n^i_1}}, \dots, \singleton{\sem{n^i_{j_i}}}}).}
            }
        }
        
    \end{enumerate}
\end{proof}

Similar to the case for may evaluation, to prove adequacy for the upper powerdomain we need a stratified version of must evaluation that is decidable. 

\begin{defn}
    Define the \emph{stratified must evaluation relation} $- \Downarrow^{\forall_-} -$ on $\tm{\F{\kw{nat}}} \times \Nat \times K(\tm{\kw{nat}})$ as the least relation satisfying the following.
    \begin{enumerate}
        \item We have $\ret{n} \Downarrow^{\forall_0} \MkSet{n}$.
        \item If $M \nextS{} \MkSet{M_1, \dots, M_m}$ and $M_i \Downarrow^{\forall_k} S_i$ for all $i \in \MkSet{1, \dots, m}$, then $M \Downarrow^{\forall_{k+1}} \bigcup_i S_i$.
    \end{enumerate}
\end{defn}

\begin{prop}\label[prop]{prop:must-stratified}
    The proposition $M \Downarrow^{\forall_k} S$ is decidable. Moreover, we have $\exists k :\Nat.~M \Downarrow^{\forall_k} S$ if and only if $M \Downarrow^{\forall} S$.
\end{prop}

\begin{restatable}[Adequacy]{thm}{ThmAdequacyUpper}\label[thm]{thm:adequacy-upper}
    If $\sem{M}\Supp{}$ holds then $M \must{} S$ for a unique finite set $S = \MkSet{n_1, \dots, n_k}$.
\end{restatable}

\begin{proof}
    By \cref{thm:coherence} it suffices to show that $\kw{run}(M)\Supp{}$ implies there is a finite set of numbers $S$ such that $M \must{} S$. We proceed by fixed-point induction. By \cref{prop:must-prop,prop:must-stratified} we have that $\Sigma_{S : K(\tm{\kw{nat}})}. M\must{} S$ is equivalent to the proposition $\exists S : K(\tm{\kw{nat}}). \exists k :\Nat. M \Downarrow^{\forall_k} S$; observing that $K(\tm{\kw{nat}}) \cong \Nat$ we have that it is also an \I{}-proposition by \cref{ax:cdj}. Consider the following subset $\Phi \subseteq (\tm{\F{\kw{nat}}} \to \UP(L \NatP))$. 

    \iblock{
        \mrow{\Phi = \MkSet{f \mid \forall M:\tm{\F{\kw{nat}}}.~f(M)\Supp \to \Sigma_{S : K(\tm{\kw{nat}})}.~M\must{} S}}
    }

    Observe that $\Phi$ may be expressed as an intersection as follows
    
    \iblock{
        \mrow{\bigcap_{M : \tm{\F{\kw{nat}}}}\MkSet{f \mid f(M)\Supp \to \Sigma_{S : K(\tm{\kw{nat}})}.~M\must{} S}}
     } 
     
     and that $\MkSet{f \mid f(M)\Supp}$ is an \I{}-subset of $\tm{\F{\kw{nat}}} \to \UP(L \NatP)$ not containing the least element. Then by \cref{cor:intersection,prop:admissible} we have that $\Phi$ is admissible. It remains to show that $F_\kw{run} \Phi \subseteq \Phi$. Let $(F_\kw{run}~f~M)\Supp$. We case on $\kw{next}(M)$. 
    \begin{enumerate}
        \item We have that $\kw{next}(M) = \varnothing$ so that $M = \ret{n}$ for some $n : \tm{\kw{nat}}$. Then we have $\ret{n} \must{} \MkSet{n}$.
        \item We have that $\kw{next}(M) = \MkSet{M_1, \dots, M_k}$ is nonempty. Then $(F_\kw{run}~f~M)\Supp = (f~M_1 \vee \dots \vee f~M_k)\Supp = (f~M_1)\Supp \land \dots \land (f~M_k)\Supp$, whence there exists $S_i$ such that $M_i \must{} S_i$ for all $i \in \MkSet{1, \dots, k}$ by induction, and the result holds by taking $S = \bigcup_i S_i$. 
    \end{enumerate}
\end{proof}

\subsection{Adequacy for the convex powerdomain construction}

The termination behavior of an element $S : \CP(L X)$ of the convex powerdomain is a combination of that of the lower and upper powerdomain: $S$ may possibly contain divergent elements or denote a totally defined element of $\CP(X)$. Because $\CP(X)$ is the \emph{free} semilattice on $X$ it supports both a defined-membership relation $\in? : X \to \CP{}(L X) \to \I{}$ like the lower powerdomain and a total termination support $\Supp{} : \CP(L X) \to \I$ like the upper powerdomain. As in the case of the upper powerdomain we have an isomorphism $\CP(X) \to \MkSet{S : \CP(L X) \mid ~S\Supp{}}$ of semilattices that we will implicitly use in the following theorems.

\begin{thm}[Soundness]\label[thm]{thm:soundness-convex}
    For all $M : \tm{\F{\kw{nat}}}$ and $n : \tm{\kw{nat}}$, we have that $M \Downarrow n$ implies $\sem{n} \mathrel{\in?} \sem{M}$ and $M \must{} \MkSet{n_1, \dots, n_k}$ implies $\sem{M}\Supp{}$ and $\sem{M} = \singleton{\sem{n_1}} \vee \dots \vee \singleton{\sem{n_k}}$.
\end{thm}

\begin{thm}[Adequacy]\label[thm]{thm:adequacy-convex}
    For all $M : \tm{\F{\kw{nat}}}$ and $u : \NatP$, $u \mathrel{\in?} \sem{M}$ implies there exists some $n : \tm{\kw{nat}}$ such that $M \Downarrow n$ and $\sem{n} = u$; moreover, if $\sem{M}\Supp{}$ holds then $M \must{} S$ for a unique finite set $S$.
\end{thm}

\section{Conclusion}\label{sec:conclusion}

In this paper we have taken the first steps in developing the theory of powerdomains in synthetic domain theory and their application to denotational semantics of nondeterminism. We expect that the results obtained will provide an integrated foundation for both programming and verifying nondeterministic programs. We conclude with remarks regarding models and suggest directions for future work.

\subsection{Models}

The powerdomain constructions of \cref{sec:powerdomains} and the proof relating the denotational and observational semantics (\cref{thm:coherence}) may be carried out in an arbitrary topos \ECat{} equipped with a complete dominance $\I{}$ closed under $\bot$ and satisfying Phoa's principle; this includes both sheaf topos models~\cite{fiore-plotkin:1996,fiore-rosolini:1997,sterling-harper:2022,van-oosten:2000} and realizability models~\cite{phoa:1991,van-oosten:2000}. As for the operational semantics, \emph{soundness} theorems (\cref{thm:soundness-lower,thm:soundness-upper,thm:soundness-convex}) in general require that \I{} is closed under binary joins, while \emph{adequacy} theorems (\cref{thm:adequacy-lower,thm:adequacy-upper,thm:adequacy-convex}) require in addition countable joins. The dominance is closed under binary joins in some sheaf models~\cite{sterling-harper:2022} and countable ones when the site chosen ensures that the Yoneda embedding preserves countable coproducts~\cite{fiore-rosolini:1997:cpos}; some realizability models~\cite{van-oosten:2000} also validate the stronger property of closing the dominance under joins of countable $\I$-propositions. 

\subsection{Future work}

\subsubsection{Characterization of the path ordering}

As mentioned in \cref{sec:synthetic-powerdomains}, our work has been inspired by prior work~\cite{phoa-taylor:1990} on the characterization of the intrinsic order structure of the synthetic convex powerdomain. While \opcit{} studied the observational order as opposed to the path order considered in this paper, we conjecture that their method may be adapted to prove half of the characterization, given that the interval \I{} is sufficiently well-behaved. This means that the ``extrinsic order'' given by the lower, upper, and Egli-Milner order on the lower, upper, and convex powerdomains are each contained within the respective path orders. For the converse direction, it can be seen that for types $X$ with \I{}-equality the path order on $\LP X$ is contained in the extrinsic lower order, \ie{} containment ordering with respect to the defined-membership relation (\cref{def:def-member}). However, characterizations of the path space of these constructions in general probably require new insights, perhaps in the form of additional axiomatic assumptions on the interval object. 

\subsubsection{Synthetic \vs{} analytic domains}

As for the general relationship between analytic and synthetic domains, there appear to be two flavors of such works. On the one hand there are extensions of models of axiomatic domain theory to (sheaf) models of synthetic domain theory~\cite{fiore-plotkin:1996,sterling-harper:2022}; these works provide a way to carry out classical constructions in synthetic domain theory. In the case of powerdomains we would like to find known categories of analytic domains that our constructions extend, \eg{} Scott domains. On the other hand one may investigate the ``analytic domain theory'' of synthetic domains, such as the internal characterization of \ocpo{}s due to Fiore and Rosolini~\cite{fiore-rosolini:1997:cpos}. Powerdomains aside, this connection appears to be sorely underdeveloped, and it would greatly serve our understanding of synthetic domain theory to investigate the established concepts of analytic domain theory (directedness, algebraicity, \etc{}) within this more general setting. 

\section*{Acknowledgment}

We are grateful to the anonymous reviewers for the comments and corrections.

\bibliographystyle{entics}
\bibliography{refs}

\begin{thebibliography}{10}
\providecommand{\url}[1]{\texttt{#1}}
\providecommand{\urlprefix}{ }
\providecommand{\eprint}[2][]{\url{#2}}

\bibitem{barr-wells:1985}
Barr, M. and C.~Wells, \emph{Toposes, Triples, and Theories}, Grundlehren der mathematischen Wissenschaften, Springer-Verlag (1985), ISBN 9780387961156.
\newline\urlprefix\url{https://books.google.co.jp/books?id=e-iMPwAACAAJ}

\bibitem{sgdt:2011}
Birkedal, L., R.~E. Mogelberg, J.~Schwinghammer and K.~Stovring, \emph{First steps in synthetic guarded domain theory: Step-indexing in the topos of trees}, in: \emph{2011 IEEE 26th Annual Symposium on Logic in Computer Science}, pages 55--64 (2011).
\newline\urlprefix\url{https://doi.org/10.1109/LICS.2011.16}

\bibitem{fiore:1997}
Fiore, M.~P., \emph{An enrichment theorem for an axiomatisation of categories of domains and continuous functions}, Mathematical Structures in Computer Science \textbf{7}, pages 591--618 (1997), ISSN 1469-8072, 0960-1295.
\newline\urlprefix\url{https://doi.org/10.1017/S0960129597002429}

\bibitem{fiore-pitts-steenkamp:2021}
Fiore, M.~P., A.~M. Pitts and S.~C. Steenkamp, \emph{{Quotients, inductive types, and quotient inductive types}}, lmcs \textbf{{Volume 18, Issue 2}} (2022).
\newline\urlprefix\url{https://doi.org/10.46298/lmcs-18(2:15)2022}

\bibitem{fiore-plotkin:1996}
Fiore, M.~P. and G.~D. Plotkin, \emph{An extension of models of axiomatic domain theory to models of synthetic domain theory}, in: D.~van Dalen and M.~Bezem, editors, \emph{Computer Science Logic, 10th International Workshop, {CSL} '96, Annual Conference of the EACSL, Utrecht, The Netherlands, September 21-27, 1996, Selected Papers}, volume 1258 of \emph{Lecture Notes in Computer Science}, pages 129--149, Springer.
\newline\urlprefix\url{https://doi.org/10.1007/3-540-63172-0\_36}

\bibitem{fiore-rosolini:1997:cpos}
Fiore, M.~P. and G.~Rosolini, \emph{The category of cpos from a synthetic viewpoint}, in: S.~D. Brookes and M.~W. Mislove, editors, \emph{Thirteenth Annual Conference on Mathematical Foundations of Progamming Semantics, {MFPS} 1997, Carnegie Mellon University, Pittsburgh, PA, USA, March 23-26, 1997}, volume~6 of \emph{Electronic Notes in Theoretical Computer Science}, pages 133--150, Elsevier.
\newline\urlprefix\url{https://doi.org/10.1016/S1571-0661(05)80165-3}

\bibitem{fiore-rosolini:1997}
Fiore, M.~P. and G.~Rosolini, \emph{Two models of synthetic domain theory} \textbf{116}, pages 151--162, ISSN 0022-4049.
\newline\urlprefix\url{https://doi.org/10.1016/S0022-4049(96)00164-8}

\bibitem{frumin-etal:2018}
Frumin, D., H.~Geuvers, L.~Gondelman and N.~v.~d. Weide, \emph{Finite sets in homotopy type theory}, in: \emph{Proceedings of the 7th ACM SIGPLAN International Conference on Certified Programs and Proofs}, CPP 2018, page 201–214, Association for Computing Machinery, New York, NY, USA (2018), ISBN 9781450355865.
\newline\urlprefix\url{https://doi.org/10.1145/3167085}

\bibitem{gunter-etal:1989}
Gunter, C.~A., P.~D. Mosses and D.~S. Scott, \emph{Semantic {Domains} and {Denotational} {Semantics}} .

\bibitem{hennessy-plotkin:1979}
Hennessy, M. C.~B. and G.~D. Plotkin, \emph{Full abstraction for a simple parallel programming language}, in: J.~Bečvář, editor, \emph{Mathematical {Foundations} of {Computer} {Science} 1979}, pages 108--120, Springer, Berlin, Heidelberg (1979), ISBN 978-3-540-35088-0.
\newline\urlprefix\url{https://doi.org/10.1007/3-540-09526-8_8}

\bibitem{hyland:1991}
Hyland, J. M.~E., \emph{First steps in synthetic domain theory}, in: A.~Carboni, M.~C. Pedicchio and G.~Rosolini, editors, \emph{Category Theory}, pages 131--156, Springer Berlin Heidelberg, ISBN 978-3-540-46435-8.

\bibitem{katsumata:2005}
Katsumata, S.-y., \emph{A {Semantic} {Formulation} of tt-{Lifting} and {Logical} {Predicates} for {Computational} {Metalanguage}}, in: D.~Hutchison, T.~Kanade, J.~Kittler, J.~M. Kleinberg, F.~Mattern, J.~C. Mitchell, M.~Naor, O.~Nierstrasz, C.~Pandu~Rangan, B.~Steffen, M.~Sudan, D.~Terzopoulos, D.~Tygar, M.~Y. Vardi, G.~Weikum and L.~Ong, editors, \emph{Computer {Science} {Logic}}, volume 3634, pages 87--102, Springer Berlin Heidelberg, Berlin, Heidelberg (2005), ISBN 978-3-540-28231-0 978-3-540-31897-2. Series Title: Lecture Notes in Computer Science.
\newline\urlprefix\url{https://doi.org/10.1007/11538363_8}

\bibitem{kock:1991}
Kock, A., \emph{Algebras for the partial map classifier monad}, in: A.~Carboni, M.~C. Pedicchio and G.~Rosolini, editors, \emph{Category Theory}, pages 262--278, Springer Berlin Heidelberg, Berlin, Heidelberg (1991), ISBN 978-3-540-46435-8.

\bibitem{levy:2001}
Levy, P.~B., \emph{Call-By-Push-Value}, Ph.D. thesis (2001).
\newline\urlprefix\url{https://www.cs.bham.ac.uk/~pbl/papers/thesisqmwphd.pdf}

\bibitem{matache-moss-staton:2022}
Matache, C., S.~Moss and S.~Staton, \emph{Concrete categories and higher-order recursion: With applications including probability, differentiability, and full abstraction}, in: \emph{Proceedings of the 37th Annual ACM/IEEE Symposium on Logic in Computer Science}, LICS '22, Association for Computing Machinery, New York, NY, USA (2022), ISBN 9781450393515.
\newline\urlprefix\url{https://doi.org/10.1145/3531130.3533370}

\bibitem{milne-milner:1979}
Milne, G. and R.~Milner, \emph{Concurrent processes and their syntax}, J. ACM \textbf{26}, page 302–321 (1979), ISSN 0004-5411.
\newline\urlprefix\url{https://doi.org/10.1145/322123.322134}

\bibitem{modelberg-vezzosi:2021}
Møgelberg, R. and A.~Vezzosi, \emph{Two guarded recursive powerdomains for applicative simulation}, Electronic Proceedings in Theoretical Computer Science \textbf{351}, pages 200--217 (2021).
\newline\urlprefix\url{https://doi.org/10.4204/EPTCS.351.13}

\bibitem{nakano:2000}
Nakano, H., \emph{A modality for recursion}, in: \emph{Proceedings of the Fifteenth Annual IEEE Symposium on Logic in Computer Science}, pages 255--266, IEEE Computer Society, ISSN 1043-6871.
\newline\urlprefix\url{https://doi.org/10.1109/LICS.2000.855774}

\bibitem{niu-sterling-harper:2024}
Niu, Y., J.~Sterling and R.~Harper, \emph{Cost-sensitive computational adequacy of higher-order recursion in synthetic domain theory} (2024). 40th Conference on Mathematical Foundations of Programming Semantics (MFPS XXXX), \eprint{2404.00212}.
\newline\urlprefix\url{https://arxiv.org/abs/2404.00212}

\bibitem{phoa:1991}
Phoa, W., \emph{Domain Theory in Realizability Toposes}, Ph.D. thesis, University of Edinburgh.

\bibitem{phoa:1994}
Phoa, W., \emph{From {Term} {Models} to {Domains}}, Information and Computation \textbf{109}, pages 211--255 (1994), ISSN 0890-5401.
\newline\urlprefix\url{https://doi.org/https://doi.org/10.1006/inco.1994.1017}

\bibitem{phoa-taylor:1990}
Phoa, W. and P.~Taylor, \emph{The {Synthetic} {Plotkin} {Powerdomain}}  (1990).

\bibitem{plotkin:1976}
Plotkin, G.~D., \emph{A {Powerdomain} {Construction}}, SIAM Journal on Computing \textbf{5}, pages 452--487 (1976). \_eprint: https://doi.org/10.1137/0205035.
\newline\urlprefix\url{https://doi.org/10.1137/0205035}

\bibitem{plotkin:1981}
Plotkin, G.~D., \emph{A structural approach to operational semantics}, Technical Report DAIMI FN-19, University of Aarhus (1981).
\newline\urlprefix\url{http://citeseer.ist.psu.edu/plotkin81structural.html}

\bibitem{pugh-sterling:2025}
Pugh, L. and J.~Sterling, \emph{When is the partial map classifier a sierpiński cone?}, in: \emph{2025 40th Annual ACM/IEEE Symposium on Logic in Computer Science (LICS)}, pages 718--731 (2025).
\newline\urlprefix\url{https://doi.org/10.1109/LICS65433.2025.00060}

\bibitem{reus:1995}
Reus, B., \emph{Program {Verification} in {Synthetic} {Domain} {Theory}}, Ph.D. thesis.

\bibitem{rijke-etal:2020}
Rijke, E., M.~Shulman and B.~Spitters, \emph{Modalities in homotopy type theory} (2020). ArXiv:1706.07526 [cs, math].
\newline\urlprefix\url{https://doi.org/10.23638/LMCS-16(1:2)2020}

\bibitem{rosolini:1986}
Rosolini, G., \emph{Continuity and effectiveness in topoi}  (1986).

\bibitem{simpson:2004}
Simpson, A., \emph{Computational adequacy for recursive types in models of intuitionistic set theory} \textbf{130}, pages 207--275, ISSN 0168-0072. Papers presented at the 2002 IEEE Symposium on Logic in Computer Science (LICS).
\newline\urlprefix\url{https://doi.org/10.1016/j.apal.2003.12.005}

\bibitem{simpson:1999}
Simpson, A.~K., \emph{Computational {{Adequacy}} in an {{Elementary Topos}}}, in: G.~Gottlob, E.~Grandjean and K.~Seyr, editors, \emph{Computer {{Science Logic}}}, Lecture {{Notes}} in {{Computer Science}}, pages 323--342, Springer, Berlin, Heidelberg (1999), ISBN 978-3-540-48855-2.
\newline\urlprefix\url{https://doi.org/10.1007/10703163_22}

\bibitem{smyth:1976}
Smyth, M.~B., \emph{Powerdomains}, in: A.~Mazurkiewicz, editor, \emph{Mathematical {Foundations} of {Computer} {Science} 1976}, pages 537--543, Springer, Berlin, Heidelberg (1976), ISBN 978-3-540-38169-3.
\newline\urlprefix\url{https://doi.org/10.1007/3-540-07854-1_226}

\bibitem{smyth:1983}
Smyth, M.~B., \emph{Power domains and predicate transformers: {A} topological view}, in: J.~Diaz, editor, \emph{Automata, {Languages} and {Programming}}, pages 662--675, Springer, Berlin, Heidelberg (1983), ISBN 978-3-540-40038-7.
\newline\urlprefix\url{https://doi.org/10.1007/BFb0036946}

\bibitem{sterling:2024}
Sterling, J., \emph{Tensorial structure of the lifting doctrine in constructive domain theory} (2024). \eprint{2312.17023}.
\newline\urlprefix\url{https://arxiv.org/abs/2312.17023}

\bibitem{sterling-gratzer-birkedal:2023}
Sterling, J., D.~Gratzer and L.~Birkedal, \emph{Denotational semantics of general store and polymorphism} (2023). ArXiv:2210.02169 [cs].
\newline\urlprefix\url{https://doi.org/10.48550/arXiv.2210.02169}

\bibitem{sterling-harper:2022}
Sterling, J. and R.~Harper, \emph{Sheaf semantics of termination-insensitive noninterference}, pages 5:1--5:19 (2022). \eprint{2204.09421}.
\newline\urlprefix\url{https://doi.org/10.4230/LIPIcs.FSCD.2022.5}

\bibitem{sterling-ye:2025}
Sterling, J. and L.~Ye, \emph{Domains and {Classifying} {Topoi}} (2025). ArXiv:2505.13096 [cs].
\newline\urlprefix\url{https://doi.org/10.48550/arXiv.2505.13096}

\bibitem{van-oosten:2000}
van Oosten, J. and A.~K. Simpson, \emph{Axioms and (counter)examples in synthetic domain theory}, Annals of Pure and Applied Logic \textbf{104}, pages 233--278 (2000), ISSN 01680072.
\newline\urlprefix\url{https://doi.org/10.1016/S0168-0072(00)00014-2}

\end{thebibliography}

\end{document}